\documentclass[11pt,a4paper]{article}
\pdfoutput=1 
\usepackage[utf8]{inputenc}
\usepackage{comment}
\usepackage{physics}
\usepackage{placeins}
\usepackage{amsmath}
\usepackage{amssymb}
\usepackage{fontawesome}
\usepackage{jheppub}
\usepackage{graphicx}
\usepackage{caption}
\usepackage{subcaption}
\usepackage[english]{babel}
\usepackage{amsthm}
\usepackage{hyperref}
\usepackage[T1]{fontenc}
\usepackage{makecell}
\usepackage{tikz}
\usetikzlibrary{calc}
\usepackage{circuitikz}

\newcommand{\be}{\begin{equation}}
\newcommand{\ee}{\end{equation}}
\newcommand{\bea}{\begin{eqnarray}}
\newcommand{\eea}{\end{eqnarray}}

\newtheorem{theorem}{Theorem}
\newtheorem{corollary}{Corollary}[theorem]

\def\tr{{\mathrm{tr}}}

\def\CC{\mathcal{C}}
\def\CD{\mathcal{D}}

\def\CF{\mathcal{F}}
\def\CG{\mathcal{G}}

\def\CI{\mathcal{I}}

\def\CN{\mathcal{N}}
\def\CO{\mathcal{O}}
\def\CP{\mathcal{P}}

\def\CW{\mathcal{W}}

\title{Thermal field theory correlators in the large-$N$ limit and the spectral duality relation}

\author[a,b]{Sa\v{s}o Grozdanov}
\author[a]{and Mile Vrbica}
\affiliation[a]{Higgs Centre for Theoretical Physics, University of Edinburgh, Edinburgh, EH8 9YL, Scotland}
\affiliation[b]{Faculty of Mathematics and Physics, University of Ljubljana, Jadranska ulica 19, SI-1000 Ljubljana, Slovenia}

\abstract{In Ref.~\cite{Grozdanov:2024wgo}, we derived a spectral duality relation applicable to the spectra of 3$d$ conformal field theories (CFTs) and their holographically dual 4$d$ black holes. In this work, we further elaborate on the properties of this duality relation and argue that the same relation can be applied to certain pairs of thermal correlator spectra in large-$N$ quantum field theories in any number of spacetime dimensions, provided the correlators are meromorphic functions with only simple poles and satisfy the thermal product formula. We discuss a rich set of properties that such retarded two-point functions must exhibit. We then show that the spectral duality relation and its implications apply to pairs of correlators in double-trace deformed CFTs and, more generally, to correlators in theories related by the Legendre transform. We illustrate, through several examples, how the spectrum of one correlator can be reconstructed from that of its dual correlation function. Notably, this includes cases relating the thermal spectra of scalar primary operators at ultraviolet and infrared fixed points, as well as current operators in a CFT$_3$ and its particle-vortex dual.
}

\begin{document}
\maketitle
\flushbottom
\section{Introduction and statement of results}
The large-$N$ limit of a quantum field theory (QFT), in which the dimension of the local Hilbert space of the theory is large, is a semi-classical limit that can be used to dramatically simplify the analysis of numerous questions in QFTs (see e.g.~Refs.~\cite{Yaffe:1981vf,Moshe:2003xn}). Beyond its practical appeal, the original expansion of various correlators in powers of `$1/N$' \cite{tHooft:1973alw} also bears a conceptually valuable connection to the topological expansion of perturbative string theory. While a number of different points of view exist that express the simplification in the $N\to\infty$ limit, none expresses its classical nature more clearly than the gauge/gravity duality (or holographic duality) (see e.g.~Ref.~\cite{Maldacena:1997re}), which can, for example, translate the problem of calculating a thermal QFT observable, which would perturbatively amount to resumming an infinite number of Feynman diagrams, to computing some geometric property of a bulk surface, solving a differential equation (see e.g.~Refs.~\cite{Zaanen:2015oix,Ammon:2015wua}) or even further reducing the calculations of correlators to only a manipulation of algebraic equations near the horizon of a bulk black hole \cite{Shenker:2013pqa,Grozdanov:2023tag}. As a result of this simplicity, holography has been invaluable in the study of generic properties of large-$N$ QFTs (mostly conformal field theories or CFTs). 

The focus of this work will be to uncover various constraints and the analytic properties of two-point correlation functions in thermal (matrix) large-$N$ strongly coupled QFTs (mostly CFTs) with holographic duals, and also to study the effects of the renormalisation group (RG) flows on such correlators. Our discussion will centre on the QFT perspective and, once the assumptions on the relevant correlators are stated, can be entirely understood within the language of QFTs without explicit reference to the gravitational dual. However, as is almost always the case in such discussions, it would have been difficult to arrive at the particular closed set of assumptions that pertain to such QFTs and it would have been extremely difficult to find interesting, computable examples of our results without the help of holography. With the help of such examples, the central theme of this work will be the field theoretic applications of the {\em spectral duality relation} that we derived in Ref.~\cite{Grozdanov:2024wgo} and discussed from the viewpoint of gravity in Ref.~\cite{Grozdanov:2025ner}. 

One specific reason to focus on retarded correlators is because they open a window into the description of thermal and time-dependent, non-equilibrium processes in QFT within the language of linear response theory. Even though our focus will be on retarded correlators, analytic continuation would allow us to make concrete statements about other types of real-time and Euclidean correlation functions. More precisely, in the thermal state with temperature $T = 1/\beta$, the causal response of an expectation value (a one-point function) of an operator $\CO(t)$ to a `small' perturbation by an external source $J(t)$ is encoded in the retarded two-point correlator $G(t)$ via
\begin{equation}
    \expval{\CO(t)}_J = \int dt' \, G(t-t') J(t') . \label{eq:green_response}
\end{equation}
In large-$N$ chaotic theories at non-zero temperature, the retarded two-point functions are believed to exhibit a universal `ringdown' form\footnote{Up to contact terms.}
\begin{equation}
    G(t)=-2i \Theta(t)\sum_n r_n e^{-i \omega_n t}, \label{eq:G_realtime}
\end{equation}
characterised by a set of complex frequencies $\qty{\omega_n}$ composed of all the quasinormal modes (QNMs) $\omega_n$ referred to as the \emph{spectrum}. In stable theories, the QNMs have $\Im \omega_n <0$ for real magnitude of the wavevector $k \equiv \sqrt{{\bf k} \cdot {\bf k}}$. In terms of the Fourier space representation with $\omega \in \mathbb{C}$, Eq.~\eqref{eq:G_realtime} is a statement that the retarded correlator $G(\omega)$ contains a possibly infinite set of poles at $\omega_n$ with respective residues $2ir_n$. It therefore decays exponentially at large times, which is directly related to the fact that its frequency-representation is assumed to contain no branch cuts. In other words, $G(\omega)$ is a meromorphic function in the complex $\omega$ plane. As discussed above, this essential property is motivated by studies of large classes of holographic models \cite{Waeber:2015oka,Grozdanov:2016vgg,Casalderrey-Solana:2018rle,Grozdanov:2018gfx,Dodelson:2023vrw,Kovtun_2005} in which the spectrum corresponds to the dual quasinormal modes of a classical asymptotically anti-de Sitter black hole and is therefore `easily' accessed by explicit computations. Importantly, the meromorphic property has also been demonstrated in a number of non-maximally chaotic large-$N$ models \cite{Choi:2020tdj,Dodelson:2024atp}. Interestingly, a counterexample to the meromorphic structure discussed here that contains branch cuts at {\em all} values of the coupling has been reported in the 3$d$ (non-holographic) large-$N$ vector model spectrum computed in Ref.~\cite{Romatschke:2019ybu}. Beyond the large-$N$ limit, one expects the emergence of branch cuts to give rise to the so-called long time tails and power-law (instead of exponential) decays of certain correlation functions in real time (see Refs.~\cite{Kovtun:2003vj,Caron-Huot:2009kyg,Grozdanov:2024fle}). Such corrections to the meromorphic large-$N$ structure, which is, as discussed, motivated by standard holographic calculations with classical bulk duals, will not be discussed in this work.

One of the goals of this work is to discuss in detail the complex-analytic properties of such correlators. Several of the results we present have been either numerically observed or suspected to be true. Here, however, we make them precise and prove them rigorously using the methods of complex analysis. We will mainly rely on two ingredients. Firstly, we use the fact that the sum in Eq.~\eqref{eq:G_realtime} must converge. This fact is intimately related to the power-law scaling of $G(\omega)$ at large $\omega$. We show that this means that the poles and zeroes of $G(\omega)$ must be evenly distributed, in the radial sense, and appear `equally frequently' in the complex $\omega$ plane. Moreover, they must lie sufficiently close to one another. Secondly, we use the thermal product formula of Ref.~\cite{Dodelson:2023vrw}, conjectured to hold generically in chaotic large-$N$ thermal theories. This will allow us to directly relate the spectrum $\qty{\omega_n}$ with the zeroes and residues of $G(\omega)$ in any number of spacetime dimensions. This result follows from a simple modification of the spectral duality relation originally derived in the context of $3d$ CFTs and used to derive constraints on the pairs of longitudinal and transverse spectra of conserved currents \cite{Grozdanov:2024wgo,Grozdanov:2025ner}.

We further apply the spectral duality relation to the study of double-trace deformed RG flows in large-$N$ theories \cite{Gubser:2002vv,Witten:2001ua,Klebanov:1999tb,Klebanov:2002ja,Faulkner:2010jy,Heemskerk:2010hk,Grozdanov:2011aa}. In particular, we show how the relation can be applied to the correlation functions of single-trace operators, with the main focus placed on such RG flows in CFTs, starting in some well-defined ultraviolet (UV) CFT. To briefly state the main results, consider a single-trace operator $\CO$ with a scaling dimension $\Delta$. We can then deform the action of a given CFT with a double-trace deformation, expressed as
\begin{equation}
    S_f = S_0- \int \frac{f}{2} \CO^2.
\end{equation}
Let $G_0(\omega)$ and $G(\omega;f)$ be the \emph{retarded} correlators of $\CO$ in the undeformed and deformed theories $S_0$ and $S_f$, respectively. The spectrum of $G(\omega;f)$, denoted by $\{\omega_n^f\}$, is the set of all points in the complex plane where $G_0(\omega)$ takes value determined by the double-trace coupling $f$,
\begin{equation}
\label{def:f_qnms}
    G_0(\omega_n^f)=-\frac{1}{f} \quad\Leftrightarrow \quad G(\omega_n^f;f) = \infty.
\end{equation}
We can now choose any pair of double-trace couplings, say $f_1$ and $f_2$, and construct a function $S(\omega)$, which contains the complete information about the spectra at both couplings:
\begin{equation}
    S(\omega)=\prod_n \qty(1-\frac{\omega}{\omega^{f_1}_n })\qty(1+\frac{\omega}{\omega^{f_2}_n}). \label{def:S_intro}
\end{equation}
This function then obeys the spectral duality relation \cite{Grozdanov:2024wgo,Grozdanov:2025ner}
\begin{equation}
    S(\omega)-S(-\omega)=2i\lambda(k; f_1, f_2) \sinh\frac{\beta\omega}{2}.
\end{equation}
The `universality' of the relation lies in the important fact that $\lambda(k; f_1, f_2)$ is \emph{not} a function of $\omega$. Moreover, it can be expressed either in terms of the spectra or the low-$\omega$ expansion of either of the two correlators. The spectral duality relation therefore presents an infinite set of stringent constraints relating the spectra of two theories related by a double-trace deformation, and, as we show below, also between the set of poles and zeroes of a meromorphic correlator. We then also use the partial fractions decomposition to show that the spectrum $\{\omega_n^{f_2}\}$ can always be derived (in practice, at least numerically) from the spectrum $\{\omega_n^{f_1}\}$, up to at most a single parameter. Concretely, we show that the spectrum $\{\omega_n^{f}\}$ is completely determined by the spectrum $\{\omega_n^0\}$ of the `same' correlator at the undeformed UV fixed point. Finally, we note that while the focus of this paper are indeed the spectra related by the double-trace deformations, the scope of applicability of the formalism developed here is broader. In particular, one can make analogous statements for the pairs of meromorphic correlators that correspond to a pair of theories related by the Legendre transform. Moreover, a number of other applications were discussed previously in our Refs.~\cite{Grozdanov:2024wgo,Grozdanov:2025ner}, where the spectral duality relation used here (i.e., the so-called self-dual limit of the full relation) was augmented by the non-trivial algebraically special structure of the dual bulk spacetimes. Besides a rich array of statements about the analytic structure of retarded correlators and its implications for the spectra related via double-trace deformations, we find it likely that the framework described here could also provide many interesting insights into the behaviour of correlators when used in conjunction with the thermal bootstrap techniques \cite{Caron-Huot:2009ypo,Iliesiu:2018fao,Petkou:2018ynm,Karlsson:2019qfi,Karlsson:2022osn,Marchetto:2023xap,Ceplak:2024bja,Buric:2025anb,Barrat:2025nvu,Buric:2025fye}.

This paper is structured as follows: In Section~\ref{sec:meromorphic}, we discuss the properties of retarded meromorphic correlators, focusing on their decompositions, and remark on the generic properties of the spectra. We derive the spectral duality relation between zeroes and poles mentioned above, and discuss the well-known cases of spectra arranged into the shape of the so-called `Christmas tree', ubiquitous in holographic theories. In Section~\ref{sec:doubletrace}, we recap the formalism of the double-trace deformed RG flow and their action on two-point correlators. We discuss the $f\rightarrow \infty$ conformal fixed point and the associated Legendre transforms. Furthermore, we interpret the motion of $\{\omega^f_n\}$ as a dynamical system, with the double-trace coupling $f$ playing the role of time. We then turn to concrete examples. In Section~\ref{sec:2d}, we focus on the holographic example of a CFT$_2$ dual to the BTZ black hole, where the correlators of scalar primary operators can be expressed in closed form, verifying many of our results explicitly. In Section~\ref{sec:4d} we turn to the thermal $\CN=4$ supersymmetric Yang-Mills (SYM) plasma, where we compute the quasinormal modes of scalar primary operators holographically. We verify the spectral duality relation numerically and present numerical methods for the computation of the spectrum at any value of the double-trace coupling, with only the knowledge of the spectrum at the fixed point. Finally, in Section~\ref{sec:3d}, we briefly discuss how the formalism of this paper can be extended to current-current correlators. We focus on the case of CFT$_3$, where we describe how the spectral duality relation relates the spectra of theories related by the particle-vortex duality or Maxwell deformations. We conclude the paper with the discussion of the generality of our results and their possible extensions. Appendix~\ref{app:RT_CT} contains a number of details of real-time correlators and contact terms. Finally, Appendices~\ref{app:complex-analysis} and \ref{app:xmas} include the relevant mathematical details with theorems and proofs that we use to make our claims in the main body of the text.

\section{Meromorphic correlators}
\label{sec:meromorphic}
We begin with a discussion of various assumptions and complex analytic properties of correlators and their spectra. The main properties of these correlators, which are relevant for the description of large-$N$ thermal field theories, are their meromorphic structure in the frequency $\omega$ plane (see e.g.~Ref.~\cite{Grozdanov:2016vgg}), their asymptotic power-law behaviour with $\omega \to \infty$, and the thermal product formula of Ref.~\cite{Dodelson:2023vrw}. 

For concreteness, we restrict our present discussion to scalar (bosonic) and Hermitian operators $\CO$ in a unitary conformal field theory (CFT). The retarded correlator of $\CO$,
\begin{equation}
    G(t,\vb x)=i \Theta(t) \expval{[\CO(t,\vb x),\CO(0,0)]}+\text{possible contact terms} ,\label{def:retarded}
\end{equation} 
describes the causal response of $\CO$ to the external perturbation of the action (cf.~Eq.~\eqref{eq:green_response}) 
\begin{equation}\label{action_perturbed}
    S\rightarrow S + \int dt d\vb{x} \, J(t,\vb{x}) \CO(t,\vb{x}) + \ldots,
\end{equation}  
where $d^d x \equiv dt d\vb{x}$. All expectation values are evaluated in the thermal state as
\begin{equation}
    \expval{\CO}=\frac{\tr\left(\CO  e^{-\beta H}\right)}{\tr\left( e^{-\beta H}\right)}, \label{eq:expval}
\end{equation}
which we assume is translationally-invariant and isotropic. The possible contact terms (a linear combination of delta functions centred at $t=0$) in Eq.~\eqref{def:retarded} come from possible higher-order terms of $J$ that might appear in the deformed action \eqref{action_perturbed}. For details, see Section~\ref{sec:3d} and Appendix~\ref{subapp:contact}.

Throughout this paper, we will be working in Fourier space, which we define as 
\begin{equation}
    G(\omega, k)=\int dt d\vb{x}\, e^{i \omega t-i \vb{k} \cdot \vb{x}}G(t,\vb{x}).
\end{equation}
Due to isotropy, the correlators will only be functions of the magnitude of the wavevector (or momentum) $k=\sqrt{\vb{k}\cdot \vb{k}}$. Note also that in most of the expressions below, we will suppress their explicit momentum dependence.

\subsection{Analytic structure and decomposition}
\label{subsec:analytic_properties}
A number of analytic properties and assumptions about retarded correlators, which we state here, are necessary and sufficient for the derivation of our results. Discussions of their standard, well-understood properties can be found, e.g., in Refs.~\cite{Bellac:2011kqa,Meyer:2011gj,Kovtun:2012rj}, while various technical details and more advanced derivations of their complex analytic properties are presented in Appendix~\ref{app:complex-analysis}.

The fact that $G(\omega)$ is a retarded correlator of a stable theory is reflected in the fact that
\begin{equation}
    G(\omega) \text{ is holomorphic for }\Im(\omega)\geq0. \label{assumption:retarded}
\end{equation}
Furthermore, hermiticity of $\CO$ gives
\begin{equation}
    G(\omega)^*=G(-\omega^*), \label{eq:hermiticity0}
\end{equation}
and
\begin{equation}
    \omega \rho(\omega) \geq 0 \text{ for } \omega \in \mathbb{R},
\end{equation}
where we have introduced the spectral density function as the analytically extended imaginary part of $G(\omega)$:
\begin{equation}
    \rho(\omega)=\frac{G(\omega)-G(-\omega)}{2i}.
\end{equation}
As mentioned above, we will be exclusively interested in correlators which are meromorphic in $\omega$, i.e., they are holomorphic in the entire complex $\omega$ plane except at isolated points. Such functions are characterised by a countable set of poles and zeroes, which we assume are simple (i.e., of order one). While higher-order zeroes and poles do appear in examples considered here as a result of collisions of two simple zeroes or poles, it is sufficient to consider them as limiting cases of the present discussion. The set of poles, which we denote here by $\qty{\omega_n^+}$, is called the \emph{spectrum} of the correlator, and is defined as
\begin{equation}
    G(\omega_n^+)=\infty.
\end{equation}
As a consequence of analyticity in the upper complex half-plane, the spectrum consists of frequencies that lie in the lower complex half-plane:
\begin{equation}
    \Im\omega^+_n<0. \label{eq:stability}
\end{equation}
The set of zeroes, which we denote for now by $\qty{\omega_n^-}$, is defined through
\begin{equation}
    G(\omega_n^-)=0.
\end{equation}
In Section \ref{sec:doubletrace}, we will see that the zeroes of $G(\omega)$ correspond to a spectrum of a correlator in a theory deformed by a double-trace operator. Both poles and zeroes lie either on the imaginary axis, or appear as mirrored pairs:
\begin{equation}
    \text{either } \Re\omega^\pm_n=0\text{ or }\omega^\pm_n=-\qty(\omega_{n'}^\pm)^* \label{eq:hermiticity}
\end{equation}

The next relevant fact comes from using the asymptotic Minkowski-space operator product expansion (OPE) techniques \cite{Dodelson:2023vrw,Caron-Huot:2009ypo}, which imply that the large-real-$\omega$ behaviour of $\rho(\omega)$ is given by
\begin{equation}
    \omega\rightarrow \infty^+: \quad \rho(\omega)\sim \omega^{2\Delta-d}, \label{assumption:OPE}
\end{equation}
where $\Delta$ is the scaling dimension of $\CO$, and $d$ is the number of spacetime dimensions. We will assume that $\rho(\omega)$ grows at most as $\omega^{2\Delta-d}$ in all directions of the complex plane, as long as that direction does not cross any poles. More precisely, we assume that
\begin{equation}
    \max_{\abs{\omega}=R_N}\abs{\rho(\omega)}< A R_N^{2 \Delta-d}, \label{assumption:rho_bounded}
\end{equation}
for some large enough $A>0$ and some sequence $R_N$ such that $R_\infty = \infty$. Beyond the asymptotic operator product expansion (OPE) expectations, we offer two additional justifications for this assumption. Firstly, $\rho(\omega)$ is expected to have a `reasonable' real-time representation, which is guaranteed by the assumption \eqref{assumption:rho_bounded}. By this statement, we mean that a well-behaved and convergent Fourier transform of $\rho(\omega)$ exists in the real-time domain, or, in other words, that $\rho(t)=\expval{[\CO(t),\CO(0)]}/2$ is finite and convergent for all $t\neq 0$, with possible contact terms or divergences at $t=0$. Secondly, in all examples known to us, this assumption is in fact satisfied. This includes a large set of holographically motivated, but very general, cases discussed in Section~\ref{subsec:christmas}. The immediate implication of assumption \eqref{assumption:rho_bounded} is the fact that we can write $\rho(\omega)$ in terms of the partial fractions decomposition \cite{Casalderrey-Solana:2018rle,Grozdanov:2025ner} (see Appendix \ref{subapp:boundedness} for details):
\begin{equation}
\label{eq:rho_partial}
    \rho(\omega)=\sum_{m=0}^{p-1} \rho_m \omega^m+\omega^p \sum_{n} \frac{r_n}{(\omega_n^+)^p}\qty[\frac{1}{\omega-\omega^+_n}+\frac{(-1)^p}{\omega+\omega_n^+}],
\end{equation}
where $r_n$ are the residues of the lower half-plane poles of $\rho(\omega)$,
\begin{equation}
    r_n=\mathrm{Res}_{\omega_n^+}\rho(\omega),
\end{equation}
$\rho_m$ are the Taylor series coefficients in the expansion of $\rho(\omega)$ around $\omega=0$,
\begin{equation}
    \rho_m = \left. \frac{\partial_\omega^m\rho(\omega)}{m!}\right|_{\omega=0},
\end{equation}
and $p$ is a positive integer satisfying
\begin{equation}
    p>2\Delta-d.\label{eq:p_assumption}
\end{equation}
If $2\Delta-d<0$, the partial fractions decomposition simplifies and we can write
\begin{equation}
    \rho(\omega)=\sum_n \frac{r_n}{\omega-\omega^+_n}+\frac{r_n}{\omega+\omega^+_n}.
\end{equation}
We note that the sums above should be interpreted as
\begin{equation}
    \sum_n = \lim_{N\rightarrow\infty} \sum_{\abs{\omega_n^+}<R_N} \label{def:sum}
\end{equation}
to guarantee convergence. 

We further assume that $G(\omega)$ exhibits power-law behaviour at large complex $\omega$, which means that $G(\omega)$ cannot grow or decay too rapidly (note that we made no assumption on the decay of $\rho(\omega)$ in the complex plane). This is motivated by the physics of Section \ref{sec:doubletrace}, where we see that, in some cases, $-1/G(\omega)$ can be interpreted as a retarded correlator in its own right, and as such, it should not diverge too rapidly in the complex plane in order to have a sensible real-time representation. Specifically, we demand that, for a sequence $R_N$ (which may be different than the one used in Eq.~\eqref{assumption:rho_bounded}), we have
\begin{equation}
    \max_{\abs{\omega}=R_N}\abs{\partial_\omega \ln G(\omega)}\leq \frac{C}{R_N},\label{assumption:G_power}
\end{equation}
for some large enough $C>0$. This assumption is further justified the fact that it holds for all cases known to us (see Sections \ref{subsec:christmas}, \ref{sec:2d} and \ref{sec:4d}), and is, as shown below, rich with information about the distribution of poles and zeroes (see Appendix \ref{subapp:powerlaw} for details). We proceed by first defining partial products over poles and zeroes:
\begin{equation}
    \pi_N^\pm(\omega) = \prod_{\abs{\omega_n^\pm}<R_N}\qty(1-\frac{\omega}{\omega_n^\pm}). \label{def:products}
\end{equation}
We can then represent $G(\omega)$ in terms of an infinite product
\begin{equation}
    G(\omega)=G(0)\lim_{N\rightarrow\infty} \frac{\pi_N^-(\omega)}{\pi_N^+(\omega)},\label{eq:G_product_decomp}
\end{equation}
where no Weierstrass canonical factors are necessary. Note that the product decomposition of Eq.~\eqref{eq:G_product_decomp} implies that the retarded correlator is, up to a scaling constant, completely determined by its poles and its zeroes (i.e., there is no multiplicative or additive freedom in its functional dependence on $\omega$). The implications of this decomposition are discussed in Section~\ref{subsec:sdr}. On the other hand, we can express $G(\omega)$ in terms of poles and their respective residues as
\begin{subequations}
\label{eq:G_partial_decomp}
\begin{equation}
    G(\omega)=g(\omega^2)+i \sum_{m=0}^{p-1}\rho_m \omega^m +2i\omega^p \sum_{n} \frac{r_n}{(\omega^+_n)^p}\frac{1}{\omega-\omega^+_n},
\end{equation}
where $g(\omega^2)$ is a real polynomial, determined by a finite number of real constants (see Appendix \ref{subapp:OPE} for details). Again, if $2\Delta-d<0$, we can simplify the expansion and write
\begin{equation}
\label{eq:G_partial2}
    G(\omega)=g(\omega^2)+2i \sum_n \frac{r_n}{\omega-\omega_n^+}.
\end{equation}
\end{subequations}
The Fourier transform of the decomposition \eqref{eq:G_partial_decomp} then gives the real-time ringdown behaviour of Eq.~\eqref{eq:G_realtime}.

We now turn to the final, holographically motivated assumption, encapsulated in the \emph{thermal product hypothesis} of Ref.~\cite{Dodelson:2023vrw} (see also Appendix \ref{subapp:TPF}). It states that, for large-$N$ chaotic theories at finite temperature, 
\begin{equation}
    \frac{\sinh\frac{\beta\omega}{2}}{\rho(\omega)} \text{ is an entire function of order one.} \label{eq:entire}
\end{equation}
In other words, it says that the function in Eq.~\eqref{eq:entire} is holomorphic everywhere and hence the only zeroes of $\rho(\omega)$ are at the integer multiples of the imaginary Matsubara frequency $2\pi i /\beta$.\footnote{We remind the reader that an entire function of order one is defined as a function holomorphic in $\mathbb{C}$ that, for large enough $\abs{\omega}$ and any $\epsilon>0$, obeys $\ln\abs{f(\omega)}<\abs{\omega}^{1+\epsilon}$. See Appendix \ref{subapp:entire} for details.} Assumption \eqref{eq:entire} then implies a trivial Hadamard decomposition of $\rho(\omega)$ known as the \emph{thermal product formula}:
\begin{equation}
    \rho(\omega)=\frac{\mu \sinh\frac{\beta\omega}{2}}{\prod\limits_n \qty[1-\qty(\frac{\omega}{\omega_n^+})^2]}. \label{eq:TPF}
\end{equation}
Here, $\mu$ is a positive and $k$-dependent `constant', related to $G(\omega)$ through
\begin{equation}
    G(\omega)=G(0)+ \frac{i}{2} \mu \beta\omega +\ldots . \label{def:mu}
\end{equation}
Finally, we can then express the residues explicitly in terms of the spectrum
\begin{equation}
    r_n=-\mu\frac{\omega^+_n \sinh\frac{\beta\omega^+_n}{2}}{2\prod\limits_{\substack{m\\ m\neq n} }\qty[1-\qty(\frac{\omega^+_n}{\omega_m^+})^2]}. \label{eq:residues}
\end{equation}
In summary, all of these considerations imply that we can express the retarded correlator $G(\omega)$ in terms of its spectrum and at most a finite number of $k$-dependent constants, as stated in Eqs.~\eqref{eq:G_partial_decomp}.

We end this section by discussing the expected validity and potential limitations of the above assumptions. As mentioned before, the assumptions are inspired by the known properties of holographic QFTs with classical bulk duals. Due to the general absence of broadly accessible correlators computed by non-holographic methods, it is difficult to go beyond our holographic intuition and make precise claims about other (matrix) large-$N$ theories with no known holographic duals, such as, e.g., the large-$N$ Yang-Mills theory (see, however, Ref.~\cite{Dodelson:2024atp} for a simpler tractable examples and discussion). Strictly speaking, the discussion in this work therefore applies to correlation functions which can be computed from a single bulk wave equation in black hole backgrounds with a non-degenerate horizon. In such cases, given very general assumptions, the correlation functions are meromorphic (see, e.g., Ref.~\cite{Gulotta_2011}), and the thermal product formula holds. For states in which there are cross-correlations between single-trace operators (i.e., where $\expval{\CO \CO'}\neq 0$ for $\CO'\neq \CO$), the holographic dual exhibits a system of coupled wave equations, rendering the situation more subtle (this was addressed with a generalisation of the thermal product formula in Ref.~\cite{Bhattacharya:2025vyi}). Ref.~\cite{Dodelson:2023vrw}, however, conjectured that the thermal product formula should hold in more general, chaotic, large-$N$ theories. This is also the point of view adopted in this paper. Nevertheless, the precise set of field theoretic conditions that would ensure the validity of the above assumptions, and therefore the results of this paper in full generality, remains unknown. 

\subsection{Poles, zeroes, and the spectral duality relation}
\label{subsec:sdr}
The analytic properties and assumptions of Section \ref{subsec:analytic_properties} provide a rich set of constraints and identities on the behaviour of zeroes and poles of $G(\omega)$. Here, we state some of the properties concerning their distributions, and constraints between them, concluding the section with the spectral duality relation \cite{Grozdanov:2024wgo,Grozdanov:2025ner}. The derivations of these statements and various technical details are presented in Appendix~\ref{app:complex-analysis}.

We begin the discussion by noting that any modification of $G(\omega)$ with a real, polynomial contact term $q(\omega^2)$, in the form of
\begin{equation}
    \widetilde G(\omega) = G(\omega) + q(\omega^2),
\end{equation}
gives another function $\widetilde G(\omega)$ which satisfies all of the assumptions of Section \ref{subsec:analytic_properties}. While poles of $G(\omega)$ and $\widetilde G(\omega)$ coincide, their zeroes, denoted by $\qty{\omega_n^-}$ and $\qty{\widetilde \omega_n^-}$, respectively, do not. Any given $\widetilde \omega_n^-$ then corresponds to a point in the complex plane where $G(\omega)$ takes a certain value:
\begin{equation}
    \widetilde G(\widetilde \omega^-_n)=0 \quad \Leftrightarrow  \quad G(\widetilde \omega^-_n) = -\left.{q(\omega^2)}\right|_{\omega=\widetilde\omega_n^-}.
\end{equation}
Crucially, this means that any property of the set of zeroes of $G(\omega)$ is really a property of \emph{any} set $\qty{\widetilde \omega^-_n}$, characterised by $q(\omega^2)$. In this work, we will mostly be interested in cases in which $q$ is a constant, meaning that we will be describing the properties of sets of points where $G(\omega)$ takes a certain real value. This is a foreshadowing of the discussion of double-trace deformations in Section~\ref{sec:doubletrace}, where $q=1/f$, and $f$ is the double-trace coupling constant. 

The partial fractions decomposition in Eq.~\eqref{eq:G_partial_decomp} tells us that the zeroes of $G(\omega)$ are completely determined by the spectrum of poles, and a finite number of real constants, namely, the coefficients of the polynomial $g(\omega^2)$. The degree of $g(\omega^2)$ is usually constrained by some physical means.\footnote{Although we do not show that in this work, the degree of $g(\omega^2)$ is expected to not exceed $2\Delta-d$ for retarded correlators that are defined canonically (see Appendix~\ref{subapp:contact}).} In fact, the physically relevant cases in Section~\ref{sec:doubletrace} and the energy-momentum tensor calculation from Ref.~\cite{Grozdanov:2025ner} require only one such free parameter, which dictates the positions of all the zeroes. While the partial fractions decomposition is extremely useful for the numerical calculations presented in Sections~\ref{sec:2d} and \ref{sec:4d}, it leaves the properties of the distribution of zeroes in the complex plane obscured.

To uncover the implications of the above assumptions and the generic properties of correlators for the spectra of poles and zeroes, we first turn to their radial distributions. Let $n_+(R)$ and $n_-(R)$ count the number of poles and zeroes of $G(\omega)$ in a disc with radius $R$, centred at $0$, respectively. The assumption on the power-law asymptotics \eqref{assumption:G_power} intuitively tells us that the majority of poles and zeroes must lie close together, effectively `cancelling' one another in the large-$\omega$ region. This can be made precise by applying the argument principle to the inequality \eqref{assumption:G_power}, arriving at (see Appendix \ref{subapp:powerlaw} for details)\footnote{In the inequality \eqref{eq:density}, the constant $C$ is the same constant as the one that appears in the inequality \eqref{assumption:G_power}.}
\begin{equation}
    \label{eq:density}
    \liminf_{R\rightarrow\infty}\abs{n_+(R)-n_-(R)}\leq C,
\end{equation}
which tells us that the radial distributions of poles and zeroes are close to one another. Furthermore, the existence of a product decomposition in the sense of Eq.~\eqref{eq:G_product_decomp} implies the convergence of the following sum:
\begin{equation}
    \lim_{N\rightarrow \infty} \abs{\sum_{\abs{\omega_n^+}<R_N}\frac{1}{\omega_n^+}- \sum_{\abs{\omega_n^-}<R_N}\frac{1}{\omega_n^-}}<\infty. \label{eq:convergence_recip}
\end{equation}
The two sums in Eq.~\eqref{eq:convergence_recip} generally fail to converge independently. The overall convergence of \eqref{eq:convergence_recip} therefore implies that the poles and zeroes must lie sufficiently close to one another and makes the claim precise.

Next, we turn to the thermal product hypothesis \eqref{eq:entire}. The fact that the canonical Hadamard product $\prod_n[1-(\omega/\omega_n^+)^2]$ needs to be a function of order one tells us that $n_+(R)$ must increase, in some sense, linearly with $R$. A standard theorem by Borel (see Appendix \ref{subapp:entire} for details) then states that this happens precisely when\footnote{In known, well-behaved examples (such as the Christmas tree configurations of Section~\ref{subsec:christmas}), stricter forms of \eqref{eq:density} and \eqref{eq:limsup} hold, namely, that $\abs{n_+(R)-n_-(R)}<D$, for some constant $D$, and $\lim_{R\rightarrow\infty}{\ln n_+(R)}/{\ln R}=1$.
In this case, it is also straightforward to show that
$\lim_{R\rightarrow\infty}{\ln n_-(R)}/{\ln R}=1.$ While these simpler, more constraining expressions are expected to hold for correlators emerging from QFT calculations, explicit counterexamples may be constructed (e.g., $\ln\ln n_+(R)={\lfloor\ln\ln R\rfloor}$ and $\ln\ln n_-(R)={\lfloor2\ln\ln R\rfloor/2}$). Such counterexamples are characterised by an uneven radial distribution of poles and zeroes, which may be `unphysical', although we are not aware of any physical principle that forbids them.
}
\begin{equation}
\label{eq:limsup}
    \limsup_{R\rightarrow\infty}\frac{\ln n_+(R)}{\ln R}=1.
\end{equation}

We can now construct an entire function $S(\omega)$ with zeroes at $\qty{\omega_n^+}$ and $\qty{-\omega_n^-}$. To do so, we use the partial products defined in Eq.~\eqref{def:products}:
\begin{equation}
    S(\omega)=\lim_{N\rightarrow\infty}\pi^+_N(\omega) \pi^-_N(-\omega), \label{def:S}
\end{equation}
which provides a convergent prescription for taking the product in Eq.~\eqref{def:S_intro}. As explained in Appendix \ref{subapp:SDR}, we can use the product decomposition of Eq.~\eqref{eq:G_product_decomp}, in conjunction with the thermal product formula of Eq.~\eqref{eq:TPF}, to show that the zeroes and the poles of $G(\omega)$ are constrained by the \emph{spectral duality relation} \cite{Grozdanov:2024wgo}:
\begin{equation}
    S(\omega)-S(-\omega)=2i\lambda \sinh\frac{\beta\omega}{2}. \label{eq:SDR}
\end{equation}
Here, we have
\begin{equation}
    \lambda=-\frac{\mu}{G(0)}, \label{def:lambda}
\end{equation}
where we emphasise that, unlike $\mu$ or $G(0)$, $\lambda$ depends only on the locations of poles and zeroes---it is a property of the spectrum. The spectral duality relation presents a stringent, infinite set of constraint equations (i.e., one for each value of $\omega$) on the poles and zeroes of $G(\omega)$.

We now make a few general remarks about the spectral duality relation \eqref{eq:SDR}. Firstly, since it  presents an infinite set of equations, we can always recover any finite part of the spectra of zeroes and poles if we know the remaining (infinite) set of them. For example, we can evaluate Eq.~\eqref{eq:SDR} at integer multiples of the imaginary Matsubara frequency, $\Omega_m=-2\pi i m/\beta$ to eliminate $\lambda$. The spectral duality relation then gives an infinite number of equations relating the spectra, one for each integer $m$:
\begin{equation}
    \prod_n \qty(1-\frac{\Omega_m}{\omega_n^+})\qty(1+\frac{\Omega_m}{\omega_n^-})=\prod_n \qty(1+\frac{\Omega_m}{\omega_n^+})\qty(1-\frac{\Omega_m}{\omega_n^-}), \quad m=\ldots,-1,0,1,\ldots.
\end{equation}
Alternatively, we can expand the spectral duality relation \eqref{eq:SDR} around $\omega=0$ to find another infinite set of identities. To proceed, we define a set
\begin{equation}
    \CW_N = \qty{\frac{1}{\omega_n^-}; \:\abs{\omega_n^-}<R_N}\cup\qty{-\frac{1}{\omega_n^+}; \:\abs{\omega_n^+}<R_N}.
\end{equation}
We can then express $\lambda$ in terms of the elementary symmetric polynomials\footnote{We remind the reader that the elementary symmetric polynomials are defined as
\begin{equation}
    e_j(\qty{w_1,w_2,\ldots, w_n})=\sum_{1\leq a_1<a_2<\ldots a_{j}\leq n}w_{a_1}w_{a_2}\ldots w_{a_j}, \nonumber
\end{equation}
and are invariant with respect to the permutations of their argument.}
\begin{equation}
    \frac{i\lambda}{(2j+1)!}\qty(\frac{\beta}{2})^{2j+1} = \lim_{N\rightarrow\infty} e_{2j+1}(\CW_N). \label{eq:lambda_e}
\end{equation}
The convergence of the limit above is ensured by the convergence of the sum in Eq.~\eqref{eq:convergence_recip}, which corresponds to the $j=0$ case. These equations should not be merely understood as expressions for $\lambda$, but as an infinite set of equations (i.e., one for each $j$) relating the spectra of poles and zeroes, as encoded in the set $\CW$. The even part of $S(\omega)$ is fixed, up to a finite number of constants, with
\begin{equation}
    S(\omega)=-\lambda G(-\omega)\frac{\sinh\frac{\beta\omega}{2}}{\rho(\omega)},
\end{equation}
by using the partial fractions decomposition \eqref{eq:G_partial_decomp} and the expression for the residues from Eq.~\eqref{eq:residues}. 

The validity of the spectral duality relation and the statements that follow from it will be demonstrated with the closed-form retarded correlator from AdS${}_3$/CFT${}_2$ discussed in Section~\ref{sec:2d}, as well as the higher-dimensional numerical computation in AdS${}_5$/CFT${}_4$ presented in Section~\ref{sec:4d}. Moreover, in Refs.~\cite{Grozdanov:2024wgo,Grozdanov:2025ner}, we verified the spectral duality relation for the energy-momentum tensor and conserved current correlators in a holographic CFT$_3$, as well as for the self-dual correlators of the linear axion model computed in Ref.~\cite{Davison:2014lua}. 

\subsection{Generalised Christmas tree configuration \faTree}
\label{subsec:christmas}
The discussion of the powerful and constraining properties of the meromorphic large-$N$ thermal correlators above may appear somewhat technical and abstract. Thus, to illustrate their far-reaching implications in holographically-motivated setups, we consider the specific configurations of spectra that follow the pattern of a `generalised Christmas tree', where the poles of the correlator are organised into asymptotic branches. Such configurations are universal in theories with classical holographic duals \cite{Kovtun_2005,Waeber:2015oka,Grozdanov:2016vgg,Casalderrey-Solana:2018rle,Grozdanov:2018gfx,Dodelson:2023vrw} (see also Refs.~\cite{Konoplya:2011qq,Berti:2009kk,Jansen:2017oag} and \cite{Motl:2003cd,Cardoso:2004up,Natario:2004jd,Cardoso:2003cj,Konoplya:2003ii}), but may well hold in much larger classes of thermal large-$N$ theories and even beyond the maximally chaotic `large-$N$' limit (see Refs.~\cite{Choi:2020tdj,Dodelson:2024atp}). 

The first important statement is that the properties described above ensure that the directions of the asymptotic lines of zeroes and poles coincide in the complex plane with asymptotically fixed spacings between the neighbouring zeroes and poles. The relative asymptotic offset between the zeroes and poles then dictates the large-$\omega$ behaviour of $G(\omega)$ in the complex plane, which is, remarkably, the same in all directions of the complex plane. The second important result is a set of identities relating the parameters of the Christmas tree branches with the temperature and the scaling dimension of the operator, generalising the statements from Ref.~\cite{Dodelson:2023vrw}. The statements presented here are proven in Appendix~\ref{app:xmas}.

We consider cases in which the spectrum of poles and zeroes of $G(\omega)$ can be split into $M$ asymptotic branches, indexed with $j$, so that
\begin{equation}
    \qty{\omega_n^\pm; n=1,2,\ldots }=\qty{\omega_{m,j}^\pm; \: \mqty{m=1,2,\ldots \\ j=1,\ldots,M}},
\end{equation}
with
\begin{equation}
    \omega^\pm_{m,j} = d^\pm_j (m+\xi_j^\pm)+\CO(m^{-\epsilon}),\label{eq:omega_expansion}
\end{equation}
for some $\epsilon>0$. The complex parameter $d_j^\pm$ describes the asymptotic direction and spacing of the branch, and the complex parameter $\xi_j^\pm$ describes the displacement of the branch. Hermiticity of the operator $\CO$ (cf.~Eq.~\eqref{eq:hermiticity}) ensures that the branches must either be imaginary (i.e., $\Re d_j^\pm=0$ and $\Im\xi_j^\pm=0$) or must appear in mirrored pairs (i.e., $d_j^\pm=-(d_{j'}^\pm)^*$ and $\xi_j^\pm=(\xi_{j'}^\pm)^*$). The key insight into the behaviour of the asymptotic branches with respect to the assumptions outlined in Section~\ref{subsec:analytic_properties} comes from the assumption \eqref{assumption:G_power}, namely, the power-law asymptotics of $G(\omega)$. As is shown in Appendix~\ref{subapp:xmas-retarded}, the only way in which $G(\omega)$ can exhibit the asymptotic power-law behaviour in all directions of the complex plane is if the branches of zeroes and poles match in their asymptotic direction as well as spacing, i.e.,
\begin{equation}
    d_j^+=d_j^-. \label{eq:branches}
\end{equation}
With this, the convergence of the product expansion \eqref{eq:G_product_decomp} is ensured. Furthermore, the assumption about the spectral asymptotics made in our Ref.~\cite{Grozdanov:2024wgo} is thereby proven to hold for all the Christmas-tree-like spectra that have the form \eqref{eq:omega_expansion}. Note that the displacement parameters $\xi_j^\pm$ do not need to match. We introduce
\begin{equation}
    \sigma_j=\xi^-_j-\xi_j^+ ,\label{eq:displacement}
\end{equation}
where $\sigma_j$ is generally expected to be real (see examples in Sections \ref{sec:2d} and \ref{sec:4d}), and parameterises the relative offset between the asymptotic branches of poles and zeroes (see Fig.~\ref{fig:illustration}). Their sum defines the relative offset index
\begin{equation}
    \sigma = \sum_j \sigma_j \label{def:sigma},
\end{equation}
which dictates the large-$\omega$ power-law behaviour of $G(\omega)$:
\begin{equation}
    \abs{\omega}\rightarrow\infty: \quad G(\omega) \sim \omega^{-\sigma}, \quad \arg\omega \neq \arg d_j.  \label{eq:G_asymptotics}
\end{equation}
Remarkably, as soon as we demand that the leading asymptotic behaviour of $G(\omega)$ has a power-law scaling, the corresponding exponent becomes the same in all directions of the complex plane that avoid the poles. Note that some finite parts of the spectra might not naturally fit into the asymptotic branches. They must be accounted for by attaching them to an asymptotic line. For example, an addition of a single pole to the correlator amounts to the displacement $\xi_j^+\rightarrow\xi_j^+-1$ for an arbitrary branch $j$, which in turn amounts to $\sigma\rightarrow\sigma+1$ (see Fig.~\ref{fig:illustration}). This can also be understood by factoring out a single pole from the product expansion of $G(\omega)$. A similar statement holds for the spectrum of zeroes.

\begin{figure}
    \centering
    \includegraphics[width=1\linewidth]{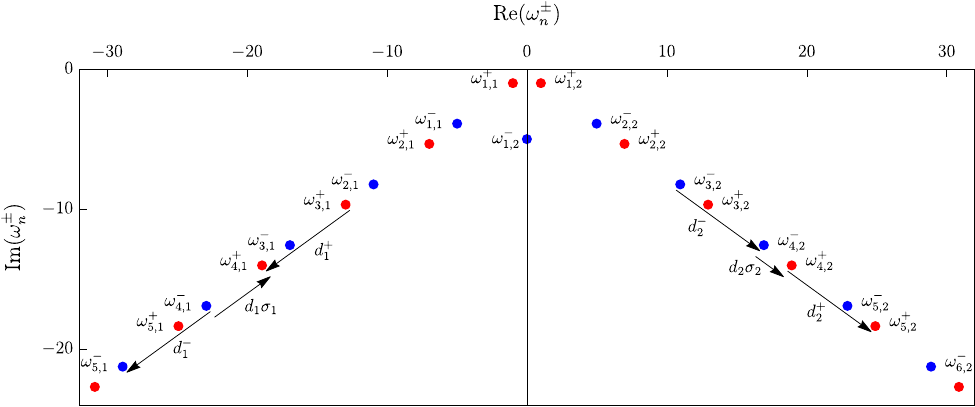}
    \caption{A toy example of a spectrum of poles (red) and zeroes (blue). We have two branches, $\omega_{m,1}$ and $\omega_{m,2}$ with, $d_1^+=d_1^-=d_1$ and $d_2^+=d_2^-=d_2$. The zero on the imaginary line is an element of $\{\omega_{m,2}\}$, which displaces its indexing relative to its mirror counterpart. The spacings $d^\pm_{j}$ are the same for the poles and the zeroes, while the relative offsets are $\sigma_1=-2/3$ and $\sigma_2=1/3$. This gives $G(\omega)\sim \omega^{-1/3}$. }
    \label{fig:illustration}
\end{figure}

Note that this analysis is not sensitive to the logarithmic corrections to the leading behaviour of the form $\omega^{-\sigma} \ln^\gamma \omega$, which do appear in specific cases of correlators (see e.g. \cite{Son:2002sd,Romatschke_2009,Casalderrey-Solana:2018rle}). This uniformity of the power-law behaviour in all directions of the complex plane and the lack of Stokes lines might be a universal feature of a more general class of meromorphic functions, described in Section \ref{sec:meromorphic}. Furthermore, it is expected that $\sigma=d-2\Delta$ for retarded correlators that are not modified by contact terms, as is shown in the examples of Sections \ref{sec:2d} and \ref{sec:4d}.

Parameters describing the spectrum of poles, namely $d_j^+$ and $\xi_j^+$ (and through Eq.~\eqref{eq:branches}, the spectrum of zeroes), are subject to constraints that follow from the properties described in Section \ref{subsec:analytic_properties}. Firstly, the retardedness property \eqref{eq:stability} gives $\Im d_j^+<0$. More interesting constraints come from the thermal product formula. Performing an asymptotic expansion (see Appendix~\ref{subapp:xmas-OPE} for details) on both sides of Eq.~\eqref{eq:TPF} (using the large-$\omega$ behaviour of $\rho(\omega)$ from Eq.~\eqref{assumption:OPE}) gives a constraint on $d_j^+$ from the leading term,
\begin{equation}
    -i\sum_{j} \frac{1}{d_j^+}=\frac{\beta}{2\pi}, \label{eq:christmas_constraint}
\end{equation}
and a constraint on $\xi_j^+$ from the subleading term,
\begin{equation}
    \sum_j 1+2\xi_j^+ = 2\Delta-d, \label{eq:xmas_constraint}
\end{equation}
generalising the results of Ref.~\cite{Dodelson:2023vrw}. It is interesting to note that the constraint \eqref{eq:christmas_constraint} has an interesting implication in cases with two branches (a single pair of branches typically seen in holographic examples) of solutions with $d_1=-d_2^*=\abs{d} e^{i\theta}$. There, we have $\beta \abs{d}=4\pi \abs{\sin\theta}$, which means that the spacing between the poles goes to zero as $\theta\rightarrow 0$. This can be understood as a formation of a branch cut along the real axis, and appears in the weak coupling limit of the holographic models, as discussed in Ref.~\cite{Grozdanov:2016vgg} (see also Ref.~\cite{Arean:2020eus,Arnaudo:2024sen}).

We can also consider branches of poles with $\omega_{m,j} \sim m^{\alpha_j}$, for some $\alpha_j \neq 1$. Such configurations appear, for example, in holographic models with a higher-derivative bulk action \cite{Dodelson:2023vrw}. In such cases, it is more difficult to make general statements. The order-one property of the thermal product formula \eqref{eq:entire} (more specifically, Eq.~\eqref{eq:limsup}) immediately tells us that we \emph{must} have at least one branch with $\alpha_j=1$, with all the other branches having $\alpha_j\geq 0$. As shown in Appendix~\ref{subapp:xmas-OPE}, the constraint \eqref{eq:christmas_constraint} still holds, with the sum only taken over branches with $\alpha_j=1$.

\subsection{Comments on sum rules and pole-skipping}
In this Section, we remark on a few additional properties and identities of meromorphic correlators that can be inferred from the above discussion. 

First, note that the decomposition \eqref{eq:G_partial_decomp} implies the convergence of the sums of the form:
\begin{equation}
    \label{eq:res_convergence}
    p>2\Delta-d: \quad \sum_n \frac{r_n}{(\omega_n^+)^{p+1}}\text{ converges},
\end{equation}
where, again, the sum must be taken in the sense of Eq.~\eqref{def:sum}. For small $p$, these sums have special significance. For operators with a low-enough scaling dimension, we have
\begin{equation}
    2\Delta - d <1: \quad \sum_n \frac{\sinh \frac{\beta \omega^+_n}{2}}{\omega^+_n \prod\limits_{\substack{m\\m\neq n}}\qty[1-\qty(\frac{\omega_n^+}{\omega_m^+})^2]} = \frac{\beta}{2}, \label{eq:sumrule}
\end{equation}
which is a recasting of the sum rule from Ref.~\cite{Grozdanov:2024wgo}. Furthermore, the spectral parameter $\lambda$ can be expressed purely in terms of the spectrum of poles, again, for a sufficiently small scaling dimension,
\begin{equation}
    2\Delta - d <0: \quad \sum_n \frac{\sinh \frac{\beta \omega^+_n}{2}}{ \prod\limits_{\substack{m\\m\neq n}}\qty[1-\qty(\frac{\omega_n^+}{\omega_m^+})^2]} = \frac{i}{2\lambda}. \label{eq:lambda}
\end{equation}
Eqs.~\eqref{eq:sumrule} and \eqref{eq:lambda} can be derived by matching the first two Taylor series coefficients of the partial fractions decomposition of $G(\omega)$ around $\omega=0$ to Eq.~\eqref{def:mu} (see Eqs.~\eqref{eq:G_partial_decomp}).

Next, we turn to the phenomenon of \emph{pole-skipping} \cite{Grozdanov:2017ajz,Blake:2018leo,Blake:2017ris} \cite{Grozdanov:2019uhi,Blake:2019otz} (see also an earlier work \cite{Amado:2008ji} and Refs.~\cite{Grozdanov:2018kkt,Grozdanov:2020koi,Wang:2022mcq,Natsuume:2019xcy,Ahn:2020baf,Blake:2021hjj,Ceplak:2019ymw,Abbasi:2020xli,Grozdanov:2023tag,Choi:2020tdj,Chua:2025vig})), which is a property of thermal correlators when their value becomes `$0/0$'. While pole-skipping is intimately connected with quantum chaos, it also has implications for the thermal spectra of holographic theories, CFTs in general, and perhaps beyond. Specifically, it can be used to reconstruct the thermal spectrum \cite{Grozdanov:2022npo,Grozdanov:2023tag}, or understand the behaviour of the momentum diffusion mode \cite{Grozdanov:2020koi,Blake:2019otz,Grozdanov:2025ner,Grozdanov:2023txs} with relatively simple computations involving only the horizon physics of a black hole.

More precisely, the pole-skipping points, denoted by $(\omega_*,k_*)$ are points in the (complex) $(\omega,k)$ space, where the lines of poles, $\omega_+(k)$, and lines of zeroes $\omega_-(k)$ intersect, as can be understood from the product decomposition of Eq.~\eqref{eq:G_product_decomp}. Specifically, we demand that 
\begin{equation}
   \omega_*= \omega_n^+(k_*)=\omega_n^-(k_*).
\end{equation}
In terms of the spectral duality relation, this gives
\begin{equation}
    S(\omega_*,k_* )=S(-\omega_*, k_*).
\end{equation}
Suppose now that the thermal product formula \eqref{eq:TPF} holds with some finite, non-zero $\lambda$. The spectral duality relation \eqref{eq:SDR} then states that the pole-skipping frequency must be a multiple of the negative imaginary Matsubara frequency, $\Omega_m=-2\pi i m/\beta$ (see Section \ref{subsec:sdr}):
\begin{equation}
    \omega_* = \Omega_m. \label{eq:ps_matsubara}
\end{equation}
While this statement is clear in holographic terms, the applicability of the thermal product formula might translate it to a broader spectrum of theories. 

An interesting violation of Eq.~\eqref{eq:ps_matsubara} appeared in Ref.~\cite{Choi:2020tdj}, which treats the Sachdev-Ye-Kitaev (SYK) chain at finite coupling and temperature. There, the authors derive an explicit form of the retarded energy-energy correlator, and find pole-skipping frequencies at non-integer multiples of the Matsubara frequency. Nevertheless, Ref.~\cite{Dodelson:2023vrw} found that the thermal product formula remains valid in that case. The resolution lies in the fact that, for this SYK chain, $\lambda$ tends either to zero or infinity at the pole-skipping values of momentum. The argument above does therefore not apply. It would be interesting to understand the pole-skipping discussion and possible violations of Eq.~\eqref{eq:ps_matsubara} in terms of the recently proposed modification of the thermal product formula \cite{Bhattacharya:2025vyi} (see also the discussion in Section~\ref{sec:conclusion}).

\section{CFTs related by double-trace deformations}
\label{sec:doubletrace}

In Section~\ref{subsec:sdr}, we saw that the spectral duality relation applies to poles and zeroes of any retarded correlator that has the properties required by the thermal product hypothesis of Ref.~\cite{Dodelson:2023vrw}. On the other hand, in the original context of Ref.~\cite{Grozdanov:2024wgo}, the spectral duality relation was derived for the poles (QNMs) of pairs of S-dual correlators in AdS${}_4$/CFT${}_3$. Here, we show how the spectral duality relation also applies much more broadly to meromorphic retarded correlators related by double-trace deformations in large-$N$ QFTs in any dimension. Specifically, in this Section, we focus on double-trace deformed CFTs. Since all the necessary facts about such CFTs are well known, we will only state the main results. For the literature on this topic, see Refs.~\cite{Gubser:2002vv,Witten:2001ua,Klebanov:1999tb,Klebanov:2002ja,Papadimitriou:2007sj,Faulkner:2010jy,Heemskerk:2010hk,Grozdanov:2011aa}. As we will see below, one of the central results of this Section is a relation between the spectra of the ultraviolet (UV) and the infrared (IR) CFTs (fixed points) that are related by an RG flow triggered by a double-trace deformation, as well as a relation between the spectra at any point along the flow. 

\subsection{Double-trace deformations and the spectral duality relation}
\label{subsec:double_trace}
Consider a conformal field theory in the large-$N$ limit in any number of spacetime dimensions $d$. For simplicity and concreteness, we will limit our discussion to scalar primary operators (see Section~\ref{sec:3d} for the case of current operators). Specifically, in this Section, we consider a single-trace scalar primary operator $\CO$ with the dimension of $\CO$ above the unitarity bound $\Delta > d/2 - 1$. We write down equations schematically, ignoring real-time issues, which are deferred to Appendix \ref{app:RT_CT}. We also assume that the operator does not condense, i.e., that $\expval{\CO}_{J=0}=0$, and that we are considering the state with external sources turned off (i.e., with no single-trace deformation).

The generating function of the CFT in the presence of a source $J$ is written as 
\begin{equation}
    e^{\Gamma[J]} = \left\langle e^{\int \CO J} \right\rangle,
\end{equation}
with the expectation value computed in the thermal state, as noted by Eq.~\eqref{eq:expval}. We assume that $\Gamma[J]$ admits an expansion, which we schematically write as 
\begin{equation}
    \Gamma[J]= \int\frac{1}{2}G_0 J^2+O(J^3). \label{gamma_expansion}
\end{equation}
In this notation, $G_0=\expval{\CO\CO}$ is simply the two-point correlation function computed in the thermal state.

We now study a double-trace deformation of that CFT, namely, the partition function
\begin{equation}
    e^{\Gamma[J;f]} = \left\langle e^{\int \CO J - \frac{f}{2} \int \CO^2} \right\rangle,
\end{equation}
where $f$ is the double trace coupling constant. $\Gamma[J,f]$ admits and expansion
\begin{equation}
    \Gamma[J;f]=\int\frac{1}{2}G(f) J^2+O(J^3).
\end{equation}
Assuming the validity of the saddle-point approximation, one can compute $G(f)$ via the Hubbard-Stratonovich trick by performing the path integral
\begin{equation}
    \expval{e^{\int \CO J-\frac{f}{2}\int \CO^2}}=\int \CD \widetilde J\expval{e^{\int 
    \CO \widetilde J+\frac{1}{2f}\int \qty(J-\widetilde J)^2}}. \label{eq:Hubbard-Stratonovich}
\end{equation}
This gives $G(\omega;f)$ in terms of $G_0$:
\begin{equation}
    G(\omega;f)=\frac{G_0(\omega)}{1 + f G_0(\omega)}. \label{eq:G_f}
\end{equation}
Even though the derivation here was Euclidean, Eq.~\eqref{eq:G_f} also holds for the real-time retarded (or advanced) correlators. This is transparent in holographic systems, where the double-trace deformation corresponds to mixed boundary conditions at the anti-de Sitter boundary, preserving the ingoing nature at the horizon. A QFT derivation of this fact is provided in Appendix~\ref{subapp:trace}.

Eq.~\eqref{eq:G_f} immediately tells us that the spectrum $\{\omega_n^f\}$, obeying $G(\omega_n^f;f)=\infty$ corresponds to the solutions to $G_0(\omega_n^f)=-1/f$ (see Eq.~\eqref{def:f_qnms}). This, in turn, tells us that we can use the partial fractions decomposition \eqref{eq:G_partial_decomp} of $G_0$ and then solve for $\omega_n^f$. We implement this procedure numerically in our 4$d$ CFT example studied in Section~\ref{sec:4d}.

In these theories, it is simple to relate the correlators at any two couplings with an equation analogous to the self-duality relation of Refs.~\cite{Herzog_2007,Grozdanov:2024wgo,Grozdanov:2025ner}. Namely, we can write
\begin{equation}
    \left[ G (\omega;f_1) \Delta f + 1 \right]\left[ G (\omega;f_2) \Delta f -1\right]=-1, \label{scalarEq}
\end{equation}
where $\Delta f = f_2-f_1$. Crucially, both factors on the left-hand side of Eq.~\eqref{scalarEq} obey all the assumptions about retarded correlators in Section \ref{subsec:analytic_properties}. Consequently, the zeroes of $\qty[G(\omega;f_1)\Delta f+1]$ are the poles of $\qty[G(\omega;f_2)\Delta f-1]$, and vice versa. Hence, the discussion of Section~\ref{subsec:sdr} applies directly to these modified `correlators'. In addition, Eq.~\eqref{scalarEq} provides us with another crucial piece of information, which is the fact that the zeroes of, say, $\qty[G(\omega;f_1)\Delta f+1]$, can be interpreted as a spectrum in their own right, and must therefore lie in the lower complex half-plane. This holds for the values of the couplings which are physically justified (see Sections~\ref{sec:2d} and \ref{sec:4d}). The consequence of this is the fact that the function $S(\omega)$, defined as 
\begin{equation}
    S(\omega)=\prod_n\qty(1-\frac{\omega}{\omega_n^{f_1}})\qty(1+\frac{\omega}{\omega_n^{f_2}}), 
\end{equation}
unambiguously encodes all the information about both the spectra. The zeroes of $S(\omega)$ are naturally split into those in the lower and those in the upper half-plane. We can now write a version of the spectral duality relation \eqref{eq:SDR} as
\begin{equation}
    S(\omega)-S(-\omega)=2i\lambda_{12}\sinh\frac{\beta\omega}{2}.
\end{equation}
The parameter $\lambda_{12} = \lambda_{12}(f_1,f_2)$ is a function of two arbitrary values of the couplings $f_1$ and $f_2$:
\begin{equation}
    \lambda_{12}(f_1,f_2)=\lambda_0 \frac{(f_1-f_2)G_0(0)} {(1+f_1 G_0(0))(1+f_2 G_0(0))}, \label{eq:lambda_of_f}
\end{equation}
where $\lambda_0$ corresponds to the constant (a function of $k$) that fixes the spectral duality relation for poles and zeroes of the undeformed $G_0$ and can be expressed either via Eq.~\eqref{def:lambda}, Eq.~\eqref{eq:lambda}, or by using the original spectrum of poles via Eq.~\eqref{eq:lambda_e} (where $\omega_n^+=\omega^0_n$). Note that the behaviour of $\lambda_{12}$ in terms of the couplings $f_1$ and $f_2$ makes it transparent that the odd part of $S(\omega)$ `measures' how different the two spectra are. As is clear from the definition of $S(\omega)$, if the spectra coincide, then $S(\omega)$ is an even function, and we get $\lambda=0$.\footnote{In the zero cosmological constant case of the spectral duality relation in AdS${}_4$ discussed in Refs.~\cite{Grozdanov:2024wgo,Grozdanov:2025ner}, this case is equivalent to the isospectrality of even and odd quasinormal modes in the asymptotically flat 4$d$ Einstein gravity.}

\subsection{The fixed-point limit and the Legendre transform}
\label{subsec:UV-IR-CFT}
Suppose that the (scalar) operator $\CO=\CO_-$ has scaling dimension $d/2-1<\Delta_-<d/2$, with the corresponding undeformed retarded correlator being $G_0=G_-$. Then, the double-trace deformation $-\frac{f}{2}\CO_-^2$ is relevant, and triggers the RG flow from the UV fixed point (CFT$_-$) to an IR (CFT$_+$) fixed point with an operator $\CO_+$ with conformal dimension $\Delta_+$ and a retarded correlator $G_+$. Reaching the IR fixed point corresponds to the (appropriately regularised) $f\rightarrow \infty$ limit. The two conformal dimensions are related as $\Delta_+ = d - \Delta_-$, and the IR and UV fixed points are related by a Legendre transform of the generating functional $\Gamma[J]$.

In Eq.~\eqref{eq:Hubbard-Stratonovich}, the $f\rightarrow \infty$ limit cannot be taken directly. Instead, we can promote the source into a dynamical field (see Eq.~\eqref{eq:Hubbard-Stratonovich})
\begin{equation}
    \expval{e^{\int O_+ J_+}}_+=\int \CD J_- \expval{e^{\int \CO_- J_- +\int J_+ J_-}}_-, \label{eq:J_pathintegral}
\end{equation}
where $\expval{\ldots}_\pm$ denotes the thermal expectation value in the CFT$_\pm$. In terms of the holographic dictionary, Eq.~\eqref{eq:J_pathintegral} swaps between the Dirichlet and the Neumann boundary conditions. In terms of the generating functionals, we have 
\begin{subequations}
\begin{align}
    e^{\Gamma_+[J_+]}&=\expval{\delta(\CO_-+J_+)}_-,\\
    e^{\Gamma_-[J_-]}&=\expval{\delta(\CO_+-J_-)}_+,
\end{align}    
\end{subequations}
which is a statement about $\Gamma_-[J_-]$ and $\Gamma_+[J_+]$ being the Legendre transforms of one another. In other words, the connected generating functional of one theory is a one-particle irreducible (1PI) generating functional of the other. In the saddle point approximation, this gives a relation between the retarded correlators in both of the theories:
\begin{equation}
    G_+(\omega) G_-(\omega)=-1.\label{eq:UV_IR}
\end{equation}
The poles of one correlator again correspond to the zeroes of the other one. This immediately implies a spectral duality relation between the UV and IR CFTs, and more generally, theories related by a Legendre transformation.

We can deform either of the CFTs with a double-trace deformation
\begin{subequations}
\begin{align}
    e^{\Gamma_-[J_-;f_-]} = \left\langle e^{\int \CO_- J_- - \frac{f_-}{2} \int \CO_-^2} \right\rangle_-, \label{Zf+}\\
    e^{\Gamma_+[J_+;f_+]}= \left\langle e^{\int \CO_+ J_+ - \frac{f_+}{2} \int \CO_+^2} \right\rangle_+\label{Zf-},
\end{align}    
\end{subequations}
where the deformation is relevant in the former and irrelevant for the latter case. When the coupling constants obey
\begin{equation}
    f_+ f_- = -1, \label{eq:couplings}
\end{equation}
the generating functionals are just different formulations of the same theory, related by
\begin{equation}
    \Gamma_-[J_-;f_-] = \Gamma_+\qty[J_-/f_-;-1/f_-] + \int\frac{1}{2 f_-} J_-^2.
\end{equation}
The retarded correlators are then related via
\begin{subequations}
\label{eq:green_relation}
 \begin{align}
    G_-(\omega;-1/f_+)&=f_+\qty[f_+ G_+(\omega;f_+)-1],\\
    G_+(\omega;-1/f_-)&=f_-\qty[f_- G_-(\omega;f_-)-1], 
\end{align}   
\end{subequations}
where $G_\pm(\omega;0)\equiv G_\pm(\omega)$ are the correlators at the fixed points. This rescales the residues and modifies the contact terms of the correlators, but does not alter the spectrum. Finally, using either one of the identities \eqref{eq:green_relation} in Eq.~\eqref{scalarEq} then yields a well-defined $f_1\rightarrow 0$, $f_2\rightarrow \infty$ limit, reducing it to Eq.~\eqref{eq:UV_IR}.

It is important to note that an analogous discussion could be repeated for any pair of theories related by the Legendre transform with the corresponding correlators, spectra and zeroes obeying various (duality) relations, which are stringently constrained by the expressions discussed here. 

\subsection{Motion of the poles as a dynamical system}
\label{subsec:dynamical}
In the final part of this Section, we find it interesting to point out that the motion of poles as a function of the double-trace coupling $f$ in the complex $\omega$ plane can be understood as a dynamical system. To see this, take a derivative of Eq.~\eqref{eq:G_f} with respect to $f$, which gives a `dynamical equation' for the correlator in the space of the double-trace coupling:
\begin{equation}
    \frac{d G(\omega;f)}{df}= -G(\omega;f)^2. \label{dyn_eq}
\end{equation}

Let us now assume that we can parameterise the partial fractions expansion of $G(f)$ with $f$ in a continuous way. More specifically, we assume that $\omega_n(f) \equiv \omega_n^f$ are continuous and (at least once) differentiable functions of $f$. Then, using the partial fractions decomposition of $G(f,\omega)$ and matching the principal parts in $\omega$, we get a dynamical equation for the poles:
\begin{equation}
    \frac{d \omega_n^f}{df}=-2i r_n^f , \label{dyn_eq_pole}
\end{equation}
and by matching the leading-order terms in the Taylor expansion,
\begin{equation}
    \mu(f)=\frac{\mu(0)}{\left(1+f G_0(0)\right)^2}.
\end{equation}
One can consider the second-order term in the principal part of Eq.~\eqref{dyn_eq} to get a dynamical equation for the residues as well, but we will not pursue this direction here. The dynamical equation \eqref{dyn_eq_pole} applies to all meromorphic correlators whether or not the thermal product formula and the spectral duality relation holds. 

Eq.~\eqref{dyn_eq_pole} has an interesting interpretation in terms of the RG flow. At the UV fixed point, the residues go to zero as we move away from the complex origin, i.e., $\lim_{n\rightarrow\infty} r_n = 0$ (see Eq.~\eqref{eq:res_convergence}). This means that a small relevant deformation most affects the poles close to the origin (the IR low-energy excitations), while leaving the UV excitations largely intact. This makes sense from the RG perspective. Moreover, we also see that an irrelevant deformation affects the UV poles much more significantly than the poles in the IR part of the spectrum, since in that case, the residues are expected to diverge with increasing $n$. This is illustrated graphically in Sections~\ref{sec:2d} and \ref{sec:4d}.

Finally, assuming the thermal product formula \eqref{eq:TPF} and using the residues \eqref{eq:residues} (see Appendix \ref{subapp:SDR} for details), we obtain a `time-independent' dynamical equation
\begin{equation}
    \frac{d\omega^\lambda_n}{d\lambda} = i\frac{\omega^\lambda_n \sinh\frac{\beta\omega^\lambda_n}{2}}{\prod\limits_{\substack{m\\ m\neq n} }\qty[1-\qty(\frac{\omega^\lambda_n}{\omega^\lambda_m})^2]}, \label{eq:dynamical_eq}
\end{equation}
where the role of time is played by $\lambda=\lambda_{12}(0,f)$, as expressed by Eq.~\eqref{eq:lambda_of_f}. This dynamical equation can be seen as the generator of spectra that obey the spectral duality relation relative to the initial (undeformed) spectrum. Of course, the system of equations remains extremely complicated due to the fact that it couples all of the (infinitely many) poles in the spectrum. 

We note that the derivation assumes isolated simple poles only, which means that the collisions among poles will likely complicate the analysis. We also remark on the possibility of subtleties related to new poles `coming in' from complex infinity. An example can be seen in Fig.~\ref{fig:worms_2d} if we adopt the perspective of an IR fixed point undergoing a double-trace deformation.

\section{Analytic example: 2$d$ CFT and the QNMs of the BTZ black hole}
\label{sec:2d}
In this Section, we consider a thermal large-central-charge CFT$_2$ dual to the BTZ black hole. Specifically, we consider scalar primary operators dual to minimally coupled scalar fields with mass $m$ in the range $-1<m^2<0$. In this case, the retarded correlators can be computed explicitly \cite{Son:2002sd} (see also Refs.~\cite{Grozdanov:2019uhi,Blake:2019otz}). For the scalar operator $\CO_+$ in the IR CFT with scaling dimension $\Delta_+$ (i.e., conventional quantisation with the Dirichlet boundary conditions), we have, up to some normalisation,
\begin{equation}
    G_+(\omega,k)=-\frac{\Gamma\qty(\frac{1+\nu}{2}+i\beta\frac{k-\omega}{4\pi })\Gamma\qty(\frac{1+\nu}{2}-i\beta\frac{k+\omega}{4\pi })}{\Gamma\qty(\frac{1-\nu}{2}+i\beta\frac{k-\omega}{4\pi })\Gamma\qty(\frac{1-\nu}{2}-i\beta\frac{k+\omega}{4\pi })}.
\end{equation}
The UV CFT correlator of the operator $\CO_-$ with dimension $\Delta_-$ is simply $G_-=-1/G_+$, which follows from Eq.~\eqref{eq:UV_IR}. Let us parameterise the scaling dimensions as
\begin{equation}
    \Delta_\pm = 1\pm \nu, \quad 0<\nu<1.
\end{equation}
The spectra of $\CO_+$ and $\CO_-$ are then given by
\begin{subequations}
\label{eq:BTZ_spectra}
\begin{align}
    \omega^+_{m,1} &= +k - \frac{2\pi i}{\beta} (2m+\nu-1),& \omega^+_{m,2} &= -k - \frac{2\pi i}{\beta} (2m+\nu-1),\\
    \omega^-_{m,1} &= +k - \frac{2\pi i}{\beta} (2m-\nu-1),& \omega^-_{m,2} &= -k - \frac{2\pi i}{\beta} (2m-\nu-1),
\end{align}
\end{subequations}
for $m=1,2,\ldots$.

This analytic example can be explicitly used to show the agreement with all of the properties of the large-$N$ correlators described in Section~\ref{sec:meromorphic}. In terms of the Christmas tree analysis of Section \ref{subsec:christmas}, we have
\begin{subequations}
\begin{gather}
    d^+_1=d^+_2=d^-_1=d^-_2=-\frac{4\pi i}{\beta}\\
    \xi^+_1=+i\frac{\beta k}{4\pi}+\frac{\nu-1}{2} ,\qquad  \xi^+_2=-i\frac{\beta k}{4\pi}+\frac{\nu-1}{2},\\
    \xi^-_1=+i\frac{\beta k}{4\pi}-\frac{\nu+1}{2},  \qquad\xi^-_2=-i\frac{\beta k}{4\pi}-\frac{\nu+1}{2}.
\end{gather}
\end{subequations}
The offset indices introduced in Eq.~\eqref{eq:displacement} are
\begin{equation}
    \sigma_1=\sigma_2=-\nu,
\end{equation}
which gives (see Eq.~\eqref{eq:G_asymptotics})
\begin{equation}
    \abs{\omega}\rightarrow \infty: \quad  G_+(\omega)\sim \omega^{2\nu}, \quad \arg\omega \neq - \frac{\pi}{2}.
\end{equation}
In this case, we can analytically compute the asymptotic behaviour of the residues,
\begin{equation}
    r^+_{m,*}\sim m^{2\nu}, \quad r^-_{m,*} \sim m^{-2\nu},
\end{equation}
where $*=\qty{1,2}$. This means that the partial fractions decomposition of Eqs.~\eqref{eq:G_partial_decomp} converges. The product decomposition of Eq.~\eqref{eq:G_product_decomp} can be  verified via the Weierstrass decomposition of the gamma function. The function $S(\omega)$, as defined in Eq.~\eqref{def:S}, can be expressed as
\begin{equation}
    S(\omega)=\frac{\cos\qty(\pi\nu-\frac{i\beta\omega}{2})+\cosh \frac{\beta k}{2}}{\cos\pi\nu+\cosh \frac{\beta k}{2}}.
\end{equation}
This explicitly verifies the spectral duality relation \eqref{eq:SDR} with
\begin{equation}
    \lambda=\frac{\sin\pi\nu}{\cos\pi\nu+\cosh\frac{\beta k}{2}},
\end{equation}
which, as required, only depends on $k$ and $\nu$, but not on $\omega$.

Finally, we also use this opportunity to point out a few new interesting facts about the dependence of the BTZ spectrum on the double-trace coupling constant $f$. We find that the relatively featureless spectrum \eqref{eq:BTZ_spectra} at the conformal fixed points (all of its QNMs are evenly spaced, and are linear functions of $k$ without any level-crossing \cite{Grozdanov:2019uhi}) exhibits much more interesting behaviour when subjected to double-trace deformations. To show this, we study the solutions to the equation
\begin{equation}
    G_-(\omega_n^{f_-})=-\frac{1}{f_-} .\label{eq:BTZ_trace}
\end{equation}
Fig.~\ref{fig:worms_2d} depicts the running of the spectrum from $\{\omega_n^-\}$ at $f_-=0$ to $\{\omega_n^+\}$ at $f_-=\infty$. The coupling constant $f_-$ must be positive to guarantee the stability condition (see Eq.~\eqref{eq:stability}). Note that an equivalent deformation of the IR fixed point would require a negative coupling constant (see Eq.~\eqref{eq:couplings}). Interestingly, as a function of $f_-$, some poles move up and some down in the complex plane. Another interesting property of the flows is that a mirrored pair of QNMs may also escape to complex infinity. This behaviour is highly sensitive to the choice of the scaling dimension, which dictates which pair of mirrored poles will escape to infinity (see the right panel of Fig.~\ref{fig:worms_2d}). For the scaling dimensions with $\nu\gtrsim 0.64$, the pair of poles that escapes to infinity is the (lowest-lying, lowest-energy) pair that is closest to the complex origin. 

\begin{figure}[ht!]
    \centering
    \includegraphics[width=0.47\linewidth]{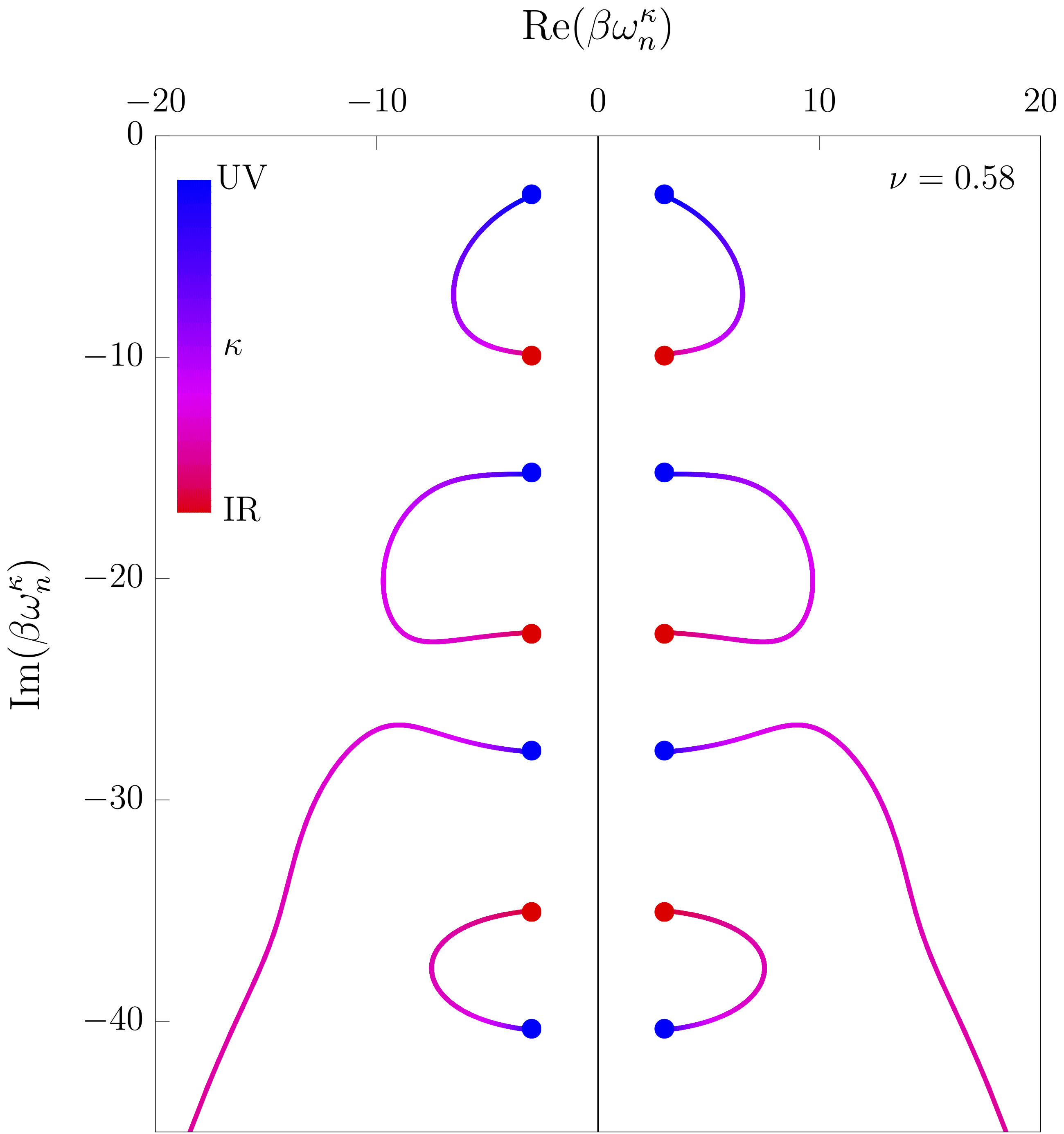}
    \includegraphics[width=0.52\linewidth]{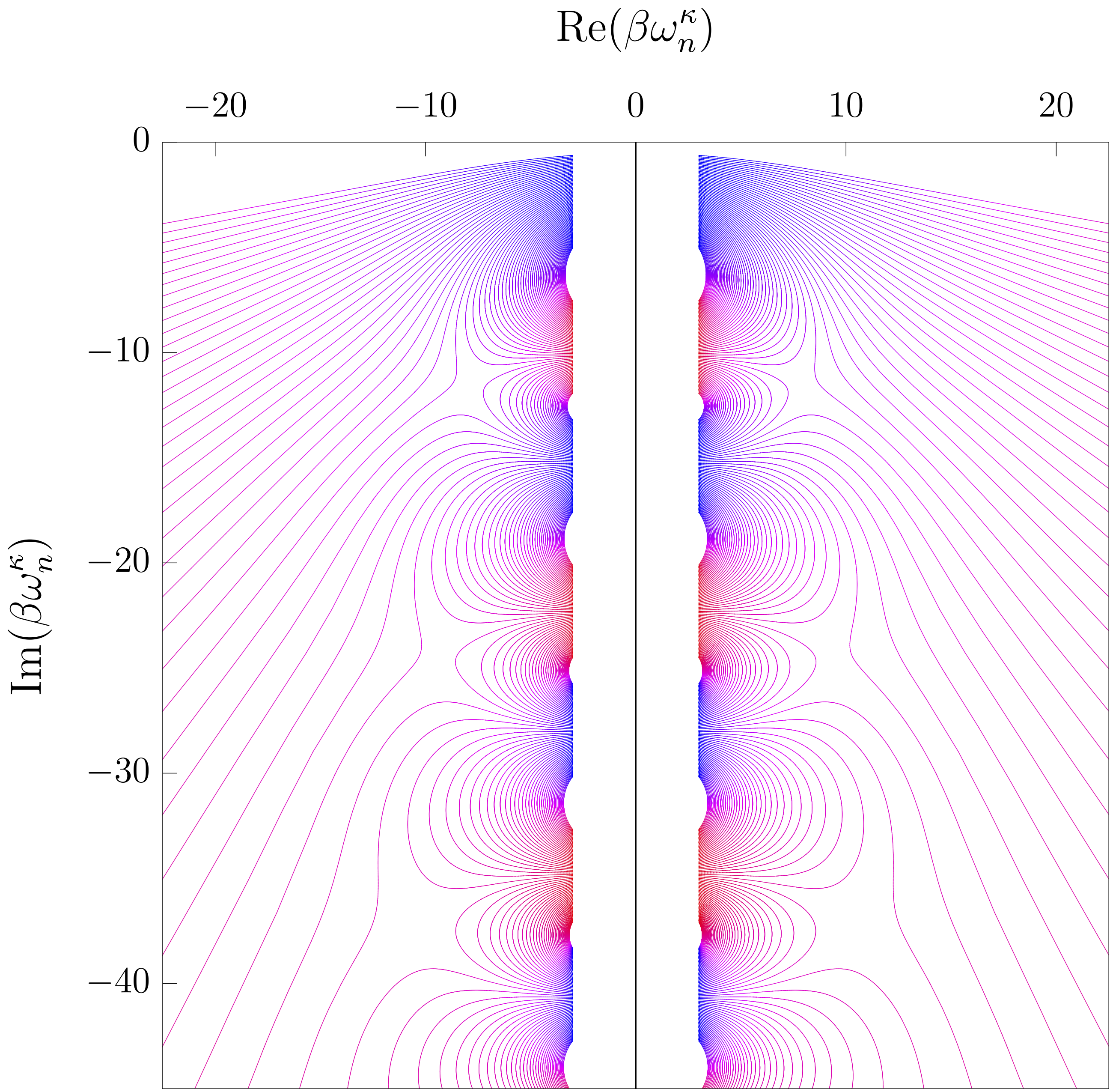}
    \caption{{\bf Left panel:} Spectrum of the retarded scalar correlator deformed by an arbitrary double-trace coupling, expressed as the full set of solutions to Eq.~\eqref{eq:BTZ_trace} at $k=3/\beta$ and $\nu=0.58$. We define $\kappa=\frac{f_--1}{f_-+1}$, so that $\kappa=-1$ at the UV and $\kappa=1$ at the IR fixed point. Along the RG flow, some poles move downwards, while some move upwards. One pair of poles `escapes to infinity'. This behaviour is highly sensitive to the scaling dimension of the scalar operator. {\bf Right panel:} Double-trace RG flows of the same scalar operators plotted at $k=3/\beta$ for different scaling dimensions with $\nu=\qty{0.2,0.21,0.22,\ldots,0.9}$. For every scaling dimension, there is a pair of poles that `escapes to infinity'. The `choice' of the pair depends sensitively on the scaling dimension.}
    \label{fig:worms_2d}
\end{figure}

\section{Numerical example: 4$d$ $\CN=4$ supersymmetric Yang-Mills theory}
\label{sec:4d}
In this Section, we consider the thermal spectra of scalar primaries in the neutral thermal state of the 4$d$ $\CN=4$ Super-Yang Mills theory at infinite 't Hooft coupling and the $N\to\infty$ limit. The spectra correspond to the QNMs of the minimally coupled massive scalar fields in the background of 5$d$ Schwarzschild-anti-de Sitter black brane (see e.g. Refs.~\cite{Kovtun_2005,Son:2002sd}).\footnote{In this Section, the spectra were computed numerically using the \texttt{QNMSpectral} library \cite{Jansen:2017oag}.} 

The goal of this Section is to verify the spectral duality relation in this, significantly more complicated setting than in a 2$d$ CFT, and present a numerical method for computing the spectrum at any given value of the double-trace coupling $f$ from the spectrum at a fixed point. This method should be understood as a demonstration of the results presented in this work. We expect that alternative approaches to the reconstructions of spectra using the spectral duality relation exist that are both more computationally efficient and physically insightful.

Let us now numerically analyse the spectrum of the operator $\CO_-$ ($\CO_+$) with scaling dimension $\Delta_-$ ($\Delta_+$) at the UV (IR) fixed point, denoted by $\qty{\omega_n^-}$ ($\qty{\omega_n^+}$). The spectrum corresponds to the QNMs of a bulk scalar field with Neumann (Dirichlet) boundary conditions. In the context of Section~\ref{subsec:christmas}, the spectra are organised into two branches with
\begin{subequations}
\label{num:asymp}
    \begin{align}
    \omega^\pm_{m,1}&=+d(m+\xi^\pm) +\delta\omega^\pm_{m}, \\ 
    \omega^\pm_{m,2}&=-(\omega^\pm_{m,1})^*,
\label{eq:sym_qnm}
\end{align}
\end{subequations}
where $\delta\omega^\pm_m$ goes to zero as $m$ goes to infinity. We pick $\Re d<0$. The constraint equations \eqref{eq:christmas_constraint} and \eqref{eq:xmas_constraint} translate into
\begin{equation}
\label{num:constraints}
    \Im\frac{1}{d}=\frac{\beta}{4\pi}, \quad \Re \xi^\pm=\frac{\pm\nu-1}{2},
\end{equation}
where we have introduced $\nu$ as
\begin{equation}
    \Delta_\pm=2\pm \nu, \quad 0<\nu<1.
\end{equation}
Eqs.~\eqref{eq:displacement}, \eqref{def:sigma} and \eqref{eq:xmas_constraint} can be used to determine the large-$\omega$ behaviour of the correlators with $\sigma = -2\nu$,
\begin{equation}
    G_\pm(\omega) \sim \omega^{\pm 2 \nu}.
\end{equation}
This can be also seen from Eq.~\eqref{eq:UV_IR} and the fact that
\begin{equation}
    \rho_\pm(\omega) \sim \omega^{\pm 2\nu}.
\end{equation}

Suppose now that we have the (numerical) knowledge of the first $m_\text{max}$ UV QNMs (with Neumann boundary conditions in the bulk), not counting their mirror images. We wish to approximate the spectrum of the correlators at the IR fixed point $\{\omega_n^+\}$, and the spectrum at some finite double-trace coupling $f$. To do so, we will compute $G_-(\omega)$ using the thermal product formula \eqref{eq:TPF} in conjunction with the partial fractions decomposition \eqref{eq:G_partial_decomp}.

Let $\omega^-_{m,1}$ correspond to the numerically computed poles in the lower-left quadrant of the complex plane, in the order of increasing moduli. We can then approximate\footnote{Note that in many cases, one can extract the parameters $d$ and $\xi$ using the WKB methods \cite{Natario:2004jd,Cardoso:2004up,Konoplya:2003ii,Cardoso:2003cj}. We do not rely on such data in this work, however, see our Ref.~\cite{Grozdanov:2025ner} for a similar computation that uses the WKB input.}
\begin{subequations}
\begin{align}
    \tilde{d} &= \qty(\omega_{m_{\text{max}},1}-\omega_{m_{\text{max}}-1,1})e^{i\theta},\\
    \tilde \xi^- &= \qty(\frac{\omega^-_{m_\text{max},1}}{\tilde d}-m_\text{max})e^{i\varphi}.
\end{align}    
\end{subequations}
Here, $\tilde d$ and $\tilde \xi$ denote the numerical approximations to $d$ and $\xi$. The phases $\theta$ and $\varphi$ are small corrections that should be set so that the constraints~\eqref{num:constraints} hold. It turns out that such corrections massively improve the approximation.

We can now use the asymptotic expansion \eqref{num:asymp} to approximate the entire spectrum with $\omega^-_{m,1}\approx\widetilde \omega^-_{m,1}$:
\begin{equation}
\label{num:approximate_spectrum}
        \widetilde\omega^-_{m,1}=\begin{cases}
            \omega^-_{m,1}, &m\leq m_\text{max},\\
            \tilde d (m +\tilde \xi^-), &m > m_\text{max}.
        \end{cases}
\end{equation}
Using the thermal product formula, we can express the approximate spectral density function as
\begin{equation}
    \rho_-(\omega;m_\text{max})=\frac{H(\omega;m_\text{max})}{H(0;m_\text{max})}\frac{\mu_-\sinh\frac{\beta\omega}{2}}{\pi^-(\omega;{m_\text{max})}},
\end{equation}
which consists of the known part of the spectrum
\begin{equation}
    \pi^-(\omega;m_\text{max})=\prod_{m=1}^{m_\text{max}}\qty[1-\qty(\frac{\omega}{\omega^-_{m,1}})^2]\qty[1-\qty(\frac{\omega}{\omega^-_{m,2}})^2],
\end{equation}
and the extrapolated part of the spectrum
\begin{equation}
    H(\omega;m_\text{max})=h(\omega;m_\text{max})h(-\omega;m_\text{max}),
\end{equation}
where we take
\begin{equation}
    h(\omega;m_\text{max})=\Gamma\qty[\frac{\omega}{\tilde d}+m_\text{max}+\tilde{\xi}^-+1]\Gamma\qty[\frac{\omega}{\tilde d^*}+m_\text{max}+(\tilde{\xi}^-)^*+1].
\end{equation}
This allows for the explicit computation of the approximate (rescaled) residues $\tilde r^-_{m,1}\approx r^-_{m,1}\mu_-$:
\begin{align}
\label{num:approximate_residues}
    m&\leq m_\text{max}:&  \tilde r^-_{m,1}&=
        - \frac{x\sinh\frac{\beta x}{2}}{2\qty[1-\qty(\frac{x}{x^*})^2]\pi_m^-(x;m_\text{max})}\frac{H(x;m_\text{max})}{H(0;m_\text{max})}\biggr|_{x=\widetilde \omega^-_{m,1}}, \\
    m&>m_\text{max}:& \tilde r^-_{m,1}&=\frac{h(x;m_\text{max})}{H(0;m_\text{max})}\frac{\sinh\frac{\beta x}{2}}{\pi^-(x;m_\text{max})}\frac{(-1)^{m+m_\text{max}}d}{(m-m_\text{max}-1)!}\nonumber \\
    & & &\times \left.\Gamma\qty[-\frac{x}{\tilde d^*}+m_\text{max}+(\tilde{\xi}^-)^*+1]\right|_{x=\widetilde \omega^-_{m,1}},\nonumber
\end{align}
where
\begin{equation}
    \pi_j^-(\omega;m_\text{max})=\prod_{\substack{m=1\\m\neq j}}^{m_\text{max}}\qty[1-\qty(\frac{\omega}{\omega^-_{m,1}})^2]\qty[1-\qty(\frac{\omega}{\omega^-_{m,2}})^2].
\end{equation}
The residues of the mirrored branch are simply
\begin{equation}
    \tilde r^-_{m,2}=(\tilde r^-_{m,1})^*.
\end{equation}
It can be shown that the large-$m$ asymptotics of the residues behave as
\begin{equation}
    m\rightarrow\infty: \quad \tilde r^-_{m,1} \sim m^{-2\nu}.
\end{equation}
This means that we can express the full retarded correlator $G_-(\omega)$ in terms of the partial fractions decomposition (see Eq.~\eqref{eq:G_partial_decomp})
\begin{equation}
\label{num:gUV}
    G_-(\omega)\approx  2i \mu_-\sum_{m=1}^M \frac{\tilde r^-_{m,1}}{\omega-\widetilde \omega_{m,1}^-}+\frac{\tilde r^-_{m,2}}{\omega-\widetilde \omega_{m,2}^-}.
\end{equation}
Here, $M$ can be chosen to be arbitrarily large. The IR spectrum can now be numerically approximated by solving $G_-(\omega_n^+)=0$, which in this procedure still depends only on the known part of the UV spectrum we initially assumed. While the computation of the spectrum at an arbitrary coupling $f^-$ requires the knowledge of $\mu_-$, this can be circumvented by absorbing $\mu_-$ into the double-trace coupling,
\begin{equation}
\label{num:swipe_UV}
    \left. 2i \sum_{m=1}^M\frac{\tilde r^-_{m,1}}{\omega-\widetilde \omega_{m,1}^-}+\frac{\tilde r^-_{m,2}}{\omega-\widetilde \omega_{m,2}^-}\right|_{\omega=\omega_n^{\mu_-f_-}} = -\frac{1}{\mu_-f_-},
\end{equation}
which gives the approximation for the double-trace deformed spectrum $\{\omega_n^{\mu_-f_-}\}$. To ensure stability, $f_-$ must be positive.

An (almost) identical procedure also allows us to approximate the UV spectrum from some finite known part of the IR spectrum. The computation of the approximate spectrum \eqref{num:approximate_spectrum} and the residues \eqref{num:approximate_residues} is completely analogous, with $-\rightarrow +$ and $\nu\rightarrow -\nu$. This gives the large-$m$ asymptotics of the residues
\begin{equation}
    \tilde r^+_{m,1} \sim m^{2\nu}.
\end{equation}
This means that we must select $p\geq 2$ in the partial fractions decomposition  \eqref{eq:G_partial_decomp}. Choosing $p=2$, this gives
\begin{equation}
\label{num:gIR}
    G_+(\omega) \approx \mu_+\qty[-\frac{1}{\lambda_+}+i\frac{\beta\omega}{2}+2i \omega^2\sum_{m=1}^M \frac{\tilde r^-_{m,1}/(\widetilde\omega_{m,1}^-)^2}{\omega-\widetilde \omega_{m,1}^-}+\frac{\tilde r^-_{m,2}/(\widetilde\omega_{m,2}^-)^2}{\omega-\widetilde \omega_{m,2}^-}].
\end{equation}
Importantly, here, we are left with an undetermined constant $\lambda_+$. Although it is likely that it is possible to determine $\lambda_+$ by other means, here, we determine it by assuming the knowledge of one pole (QNM) from the UV spectrum through the equation $G_+(\omega_n^-)=0$. It is interesting to note that there is no such ambiguity at the UV fixed point. This aligns with the commonly understood fact that in the RG flows from UV fixed point to the IR fixed point, no extra information is necessary to reproduce the IR physics from the UV physics. Similarly to the case above, we can now determine the spectra at an arbitrary (rescaled) value of the double-trace coupling through
\begin{equation}
\label{num:swipe_IR}
    \left.\qty[-\frac{1}{\lambda_+}+i\frac{\beta\omega}{2}+2i \omega^2\sum_{m=1}^M \frac{\tilde r^-_{m,1}/(\widetilde\omega_{m,1}^-)^2}{\omega-\widetilde \omega_{m,1}^-}+\frac{\tilde r^-_{m,2}/(\widetilde\omega_{m,2}^-)^2}{\omega-\widetilde \omega_{m,2}^-}]\right|_{\omega=\omega_{n}^{\mu_+f_+}}=-\frac{1}{\mu_+ f_+}.
\end{equation}
Here, $f_-$ must be negative for stability. The spectra $\{\omega_n^{\mu_+f_+}\}$ and $\{\omega_n^{\mu_-f_-}\}$ coincide when
\begin{equation}
    (f_- \mu_-)(f_+ \mu_+)=-\lambda_+.
\end{equation}

We explicitly analyse the spectra for $\nu=1/3$ and $\nu=2/3$ for $k=0$ and $k=3/\beta$. The UV and the IR spectra are computed by solving the bulk wave equation numerically with Neumann and Dirichlet boundary conditions, respectively. Eqs.~\eqref{num:swipe_UV} and \eqref{num:swipe_IR} are then used to compute the double-trace deformed spectra. By using $m_\text{max}=20$ and $M=50$, the accuracy of the first UV (IR) poles computed from the IR (UV) spectrum, compared to the ones computed from the bulk wave equation (with precision of approximately $130$ digits), falls within $2\%$. Fig.~\ref{fig:4d_worms} shows the dependence of the spectrum on the double-trace coupling. Importantly, the numerical computation correctly predicts the relative offset between the lines of poles $\sigma=-2\nu$ (see also Fig.~\ref{fig:illustration}). Note that, unlike in the BTZ case discussed in Section~\ref{sec:2d}, we do not notice the qualitative change in the motion of the poles as we changed the scaling dimension, at least for the scaling dimensions considered.

\begin{figure}
    \centering
    \includegraphics[width=0.8\linewidth]{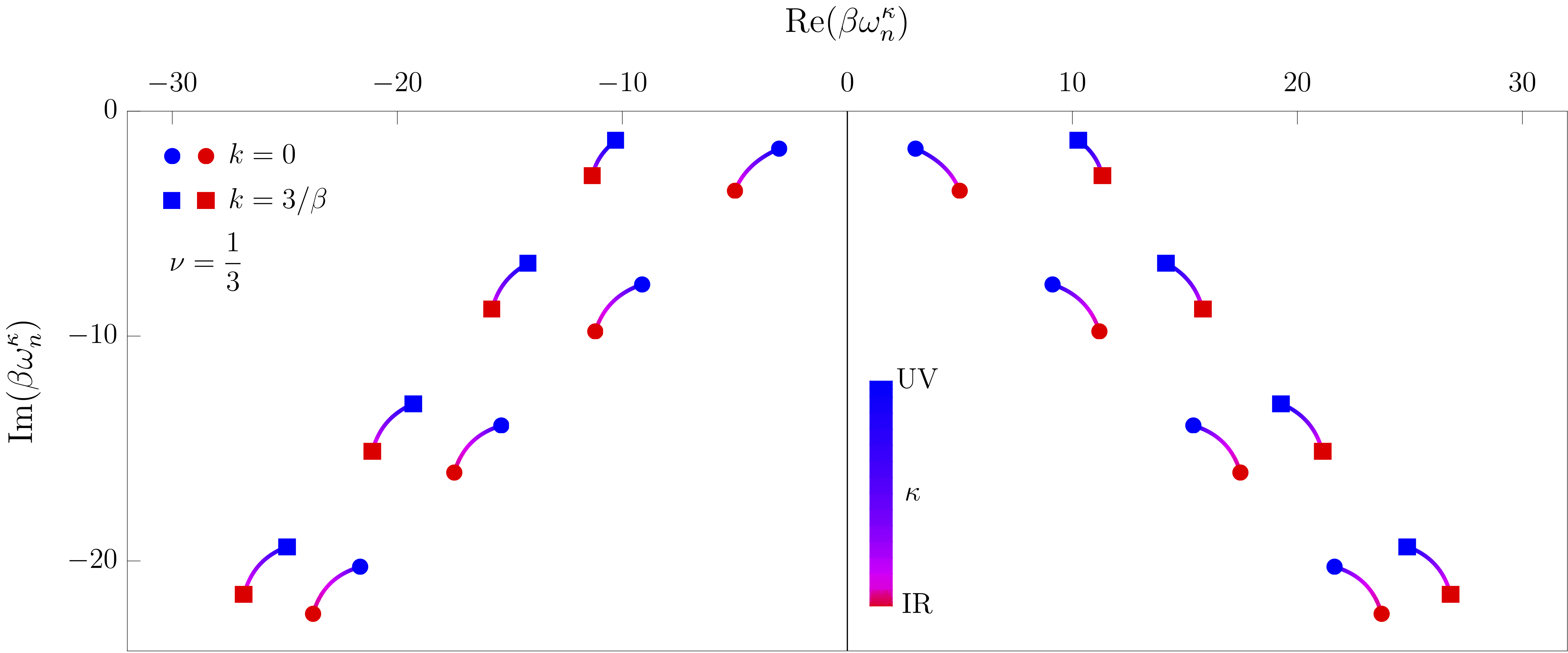}
    \includegraphics[width=0.8\linewidth]{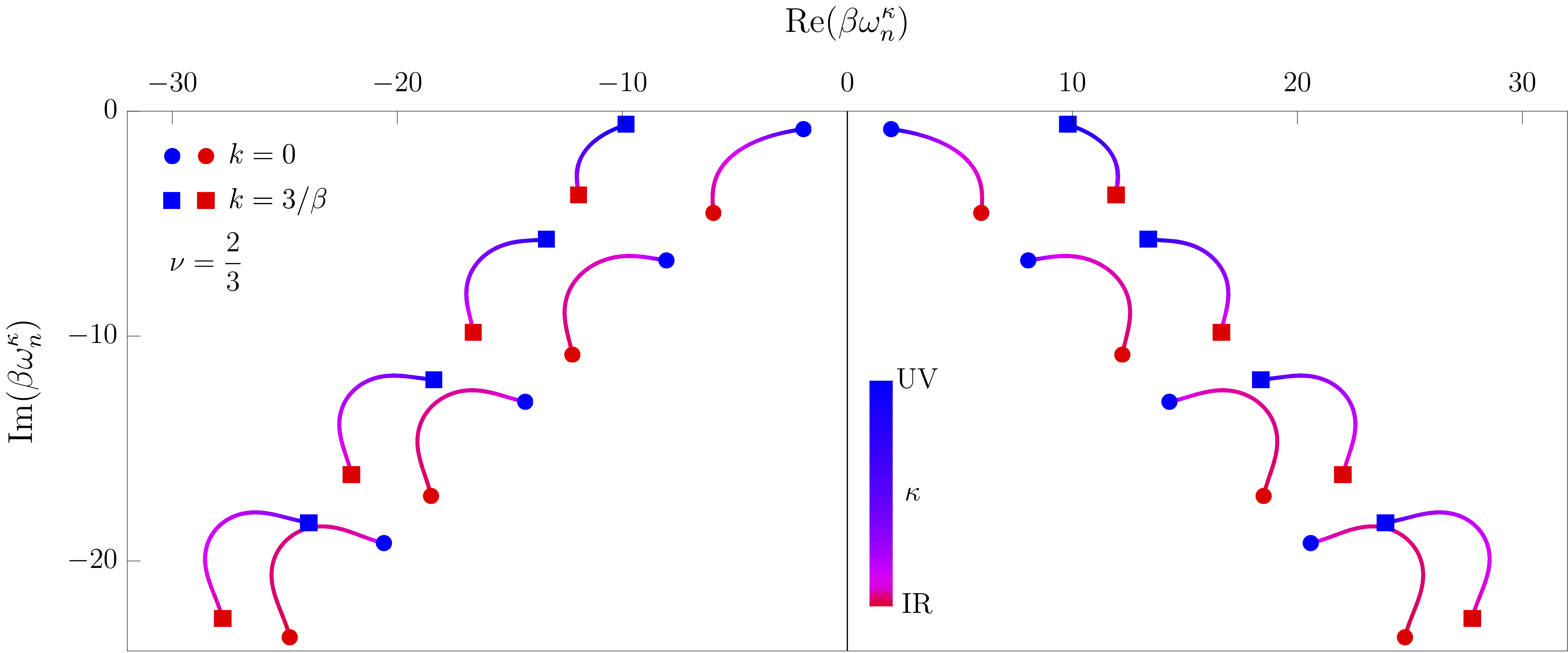}
    \caption[]{The double-trace deformed flows of the scalar spectrum between the UV to IR fixed points shown for $\nu=1/3$ (above) and $\nu=2/3$ (below). Both flows are shown for two values of momentum. Here, $\kappa$ is a rescaled coupling constant
    $\kappa=\frac{f_- G_-(0)-1}{f_- G_-(0)+1}=\frac{1-f_+ G_+(0)}{1+f_+ G_+(0)}$, such that $\kappa=-1$ at the UV and $\kappa=1$ at the IR fixed point. The spectra at all couplings were computed through the knowledge of the IR spectrum (see Eq.~\eqref{num:swipe_IR}), with the constant $\lambda_+$ set by the lowest-lying UV-fixed-point pole (i.e., by setting $G_+(\omega_{1,1}^-)=0$). The exception is the zero-momentum spectrum with $\nu=2/3$, where the spectra at various $\kappa$ were computed only with the knowledge of UV spectrum (see Eq.~\eqref{num:swipe_UV}), as that proved to be most numerically accurate. The fixed point spectra (circles and squares) were computed numerically by solving the bulk wave equation with Neumann or Dirichlet boundary conditions. We note that the trajectories that are closer to the complex origin are `more blue', meaning that they are affected more by a small deformation of the UV fixed point. As we increase the coupling, the low-lying poles then relax to their IR positions much sooner than the higher-energy QNMs.}
    \label{fig:4d_worms}
\end{figure}

We now turn to the computation of $\lambda=\lambda_+$, which appears in the spectral duality relation \eqref{eq:SDR}. It can be computed from the partial fractions decomposition, either from the UV spectrum, or from the IR spectrum and a single UV pole (see Eqs.~\eqref{def:lambda}, \eqref{eq:lambda} and \eqref{eq:G_partial_decomp}). Alternatively, it can be computed by using the elementary symmetric polynomial expressions from Eq.~\eqref{eq:lambda_e}, in terms of both the UV and IR spectra. The convergence for all methods is shown for the cases with $\nu=1/3$ and $k=0$ in Fig.~\ref{fig:lambdas}.

\begin{figure}
    \centering
    \includegraphics[width=0.7\linewidth]{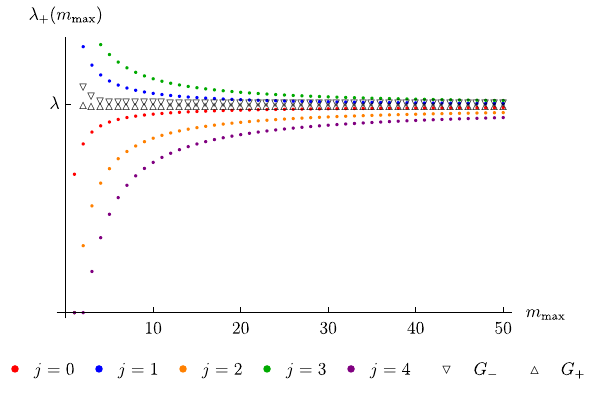}
    \caption{The parameter $\lambda=\lambda_+\approx 0.27$ computed with elementary symmetric polynomials for various values of $j$ (see Eq.~\eqref{eq:lambda_e}). This is compared with the same parameter as computed through $\lambda=-2i G_-'(0)/\beta G_-(0)$ (downwards triangles), or $\lambda=2 i G_+'(0)/\beta G_+(0)$ (upwards triangles), where $G_-(\omega)$ and $G_+(\omega)$ are expressed through Eqs.~\eqref{num:gUV} and \eqref{num:gIR} respectively, with $M=50$. The former is equivalent to Eq.~\eqref{eq:lambda}, while the latter is determined from the IR spectrum and the lowest-lying UV pole through $G_+(\omega_{1,1}^-)=0$.}
    \label{fig:lambdas}
\end{figure}

Finally, we verify the spectral duality relation \eqref{eq:SDR}. Let us define the partial product
\begin{equation}
\label{num:S}
    S(\omega,{m_\text{max}})=\prod_{m=1}^{m_\text{max}}\qty(1-\frac{\omega}{\omega^+_{m,1}})\qty(1-\frac{\omega}{\omega^+_{m,2}})\qty(1+\frac{\omega}{\omega^-_{m,1}})\qty(1+\frac{\omega}{\omega^-_{m,2}}).
\end{equation}
The odd part of $S(i\abs{\omega})-S(-i\abs{\omega})$ then converges to the sine function, as is shown in Fig.~\ref{fig:S}. We note that the convergence of the spectral duality relation is very slow, and is not improved by an extrapolation of the spectrum in the sense of Eq.~\eqref{num:approximate_spectrum}. Moreover, the convergence is much slower than in the case of the energy-momentum correlator studied in Ref.~\cite{Grozdanov:2024wgo,Grozdanov:2025ner}, which may possibly be attributed to the scalar's non-integer scaling dimension, which gives a non-integer offset index $\sigma$, meaning that the UV and the IR poles are further away from one another asymptotically.
\begin{figure}
    \centering
    \includegraphics[width=0.7\linewidth]{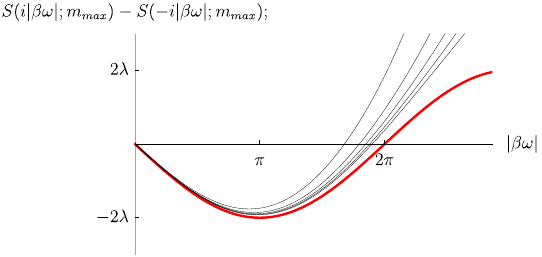}
    \caption{The convergence of the spectral duality relation \eqref{eq:SDR} with the UV and the IR spectra computed directly from the bulk. The odd part of the truncated product \eqref{num:S} is shown for $m_\text{max}=\qty{10,20,30,40,50}$. The function $-2\lambda \sin \beta\abs{\omega}/2$ is shown in red.}
    \label{fig:S}
\end{figure}

We conclude this Section by several comments about this numerical procedure. Firstly, we note that the convergence properties worsen as we increase momentum. This is expected, since the spectrum moves farther away from the `ideal' Christmas tree configuration. Secondly, picking a larger value of $p$ in the partial fractions expansion \eqref{eq:G_partial_decomp} greatly improves convergence properties, at the cost of having to supply more Taylor series coefficients of $G(\omega)$ around $\omega=0$. It is likely that a more efficient and accurate method to determine these coefficients exists beyond the one used in this work. Importantly, other numerical approaches to the problem of determining the deformed spectra are possible, and were left uninvestigated in this work. This includes a direct algebraic manipulation of the spectral duality relation \eqref{eq:SDR}, or a numerical integration of the dynamical equation \eqref{eq:dynamical_eq}. We expect that further CFT input, such as for example the OPE data, would improve the convergence of the methods presented and examined here.

\section{Beyond the scalar operators: conserved currents in CFT$_3$}
\label{sec:3d}
In this Section, we briefly comment on the generalisation of our discussion to operators with spin. For simplicity, we focus on CFTs in three spacetime dimensions with a conserved U$(1)$ current. Such CFTs are characterised by the presence of the particle-vortex duality (or S-duality), which, in the large-$N$ limit, gives rise to the spectral duality relation \cite{Grozdanov:2024wgo}. Such dualities have been explored in a variety of applications, most notably in the study of quantum critical points in condensed matter physics (see e.g.~Refs.~\cite{Son:2015xqa,Karch_2016,Seiberg_2016,Kovtun_2005,Senthil:2018cru,witten2003sl2zactionthreedimensionalconformal}). Here, we briefly discuss how the spectral duality relation arises in the cases where the assumptions of Section \ref{subsec:analytic_properties} are satisfied, we comment on the special self-dual limit of the relations and comment on Maxwell-like (`double-trace') deformations of the theory. While we do not present numerical results, we note that some of the formulae discussed here have been numerically verified in Refs.~\cite{Grozdanov:2024wgo,Grozdanov:2025ner}.

\subsection{Current-current correlators in 3$d$}
Consider a conserved U$(1)$ current $J^\mu$, sourced with an external gauge field $A_\mu$.  In linear response theory, with $G$ the associated retarded correlator of the currents, we can write
\begin{equation}
    \expval{J^\mu(\omega,k)}_A=G^{\mu\nu}(\omega,k) A_\nu(\omega,k).
\end{equation}
We assume a homogeneous, isotropic and neutral thermal state of a parity and time-reversal symmetric theory. Without loss of generality, we pick the momentum to lie in the $z$ direction, so that $k_\mu dx^\mu=-i\omega\, dt+i k\, dz$ in the $(t,y,z)$ coordinate system. Because the current is conserved, the following momentum space Ward identity must hold:
\begin{equation}
    k_\mu G^{\mu\nu}=0.
\end{equation}
We can then decompose the correlator into its longitudinal ($+$) and transverse ($-$) parts (see e.g.~Ref.~\cite{Kovtun_2005}):
\begin{equation}
    G^{\mu\nu}=G_+ P^{\mu\nu}_+ + G_- P^{\mu\nu}_-,
\end{equation}
where the two projectors are given by
\begin{equation}
    P_+^{\mu\nu} = \frac{1}{2\pi}\mqty(k^2 & 0 & k\omega \\ 0 & 0 & 0 \\ k\omega & 0 & \omega^2),\quad
    P_-^{\mu\nu} = \frac{1}{2\pi}\mqty(0 & 0 & 0\\ 0&1&0\\ 0 &0 &0).
\end{equation}

\subsection{Particle-vortex duality}
Suppose there exists a CFT$_A$ with a conserved U$(1)$ current $J^\mu_A$ sourced by an external gauge field $A$ and a generating functional $\Gamma_A[A]$. We demand that such a current is gauge invariant, i.e., rather than adding an $A_\mu J_A^\mu$ term to the action, we must gauge the derivatives. This affects the contact terms of the correlators (see e.g.~Refs.~\cite{Hartnoll:2007ip,Katz:2014rla} and Appendix~\ref{subapp:contact}). We consider a neutral state, i.e., a state with zero charge density and zero magnetic field. The generating functional then schematically admits an expansion
\begin{equation}
    \Gamma_A[A] = \int  \frac{1}{2} G_A^{\mu\nu}A_\mu A_\nu.
\end{equation}

The particle-vortex dual of the theory $A$ is derived by performing a path integral over $A$, which gives a theory with a conserved `topological' U$(1)$ current \cite{witten2003sl2zactionthreedimensionalconformal,Seiberg_2016,Leigh_2003,Karch_2016,Burgess_2001,Alejo_2019,Senthil:2018cru}
\begin{equation}
\label{def:top_current}
    J_B^\mu = \frac{1}{2\pi}\epsilon^{\mu\rho\sigma}\partial_\rho A_\sigma.
\end{equation}
This current is sourced by an external gauge field $B_\mu$
\begin{equation}
\label{def:PV}
    e^{\Gamma_B[B]}=\int \CD A e^{\Gamma_A[A]+\frac{1}{2\pi}\epsilon^{\mu\rho\sigma}\int B_\mu \partial_\rho A_\sigma},
\end{equation}
where
\begin{equation}
    \Gamma_B[B]= \int \frac{1}{2}G^{\mu\nu}_B B_\mu B_\nu.
\end{equation}
The integral in Eq.~\eqref{def:PV} precisely parallels the discussion of Section~\ref{subsec:UV-IR-CFT}. Whenever the saddle point approximation is justified (see Ref.~\cite{Marolf_2006} for a discussion in the context of holography, where this transformation swaps between magnetic and electric boundary conditions), this gives a set of duality relations between the CFT$_A$ and CFT$_B$:
\begin{subequations}
\begin{align}
    G^A_-(\omega) G^B_+(\omega) &= -1\\
    G^A_+(\omega) G^B_-(\omega) &= -1.
\end{align}
\end{subequations}
Consequently, a spectral duality relation exists between the longitudinal (transverse) spectrum of CFT$_A$ and the transverse (longitudinal) spectrum of CFT$_B$.

\subsection{Double-trace deformed CFTs and dynamical Maxwell fields}
Consider now a CFT$_A$ and a CFT$_B$, which are related by the particle-vortex duality (see Eq.~\eqref{def:PV}). Suppose we let the external field $A$ in CFT$_A$ be dynamical with explicit Maxwell dynamics,\footnote{See Refs.~\cite{Marolf_2006,Leigh_2003} for a related discussion in the context of holography, and Refs.~\cite{Grozdanov:2017kyl,Hofman:2017vwr} for similar implementations of dynamical electromagnetism in the 4$d$ boundary theory (see also Refs.~\cite{DeWolfe:2020uzb,Frangi:2024enh}).}
\begin{equation}
    \Gamma_A[A]\rightarrow \Gamma_A[A]-\frac{f}{4(2\pi)^2}\int F_{\mu\nu}F^{\mu\nu},
\end{equation}
where $F_{\mu\nu}=\partial_\mu A_\nu-\partial_\nu A_\mu$. This can be interpreted as a double-trace deformation of the CFT$_B$, 
\begin{equation}
    \int \CD A e^{\Gamma_A[A]-\frac{f}{4(2\pi)^2}\int F_{\mu\nu}F^{\mu\nu}+\frac{1}{2\pi}\epsilon^{\mu\nu\sigma}\int B_\mu \partial_\nu A_\sigma} = e^{\Gamma_B[B,f]},
\end{equation}
where $\Gamma[B;f]$ is given by
\begin{equation}
    e^{\Gamma_B[B;f]}=\expval{e^{-\frac{f}{2}\int J^\mu_B J^\nu_B \eta_{\mu\nu}+\int J_B^\mu B_\mu}}_B.
\end{equation}
The current $J_B^\mu$ is defined in Eq.~\eqref{def:top_current}, and the identity
\begin{equation}
    \frac{1}{2(2\pi)^2}F_{\mu\nu}F^{\mu\nu}=J^\mu_B J^\nu_B \eta_{\mu\nu}
\end{equation}
holds. This gives the following set of identities between the retarded correlators of CFT$_A$ and the deformed CFT$_B$:
\begin{subequations}
\begin{align}
    G^{B}_-(\omega;f) \qty(G_+^A(\omega;0)+\frac{f}{8\pi})&=-1, \label{maxwell1}\\
    G^{B}_+(\omega;f) \qty(G_-^A(\omega;0)+\frac{f}{8\pi}(\omega^2-k^2))&=-1. \label{maxwell2}
\end{align}
\end{subequations}
Since the correlators are modified by even-in-$\omega$ contact terms, the discussion of Section~\ref{sec:meromorphic} applies. This, again, gives a set of spectral duality relations that relate the spectra between the two theories.

\section{Conclusion and future directions}
\label{sec:conclusion}
In this work, we discussed the complex analytic properties of retarded correlation functions in large-$N$ quantum field theories (QFTs) at finite temperature, with most of our attention focused on conformal field theories (CFTs). We studied the properties of retarded correlators, which were assumed to be meromorphic functions in the complex $\omega$ plane. This property is ubiquitous in thermal QFTs with holographic duals where, via the holographic dictionary, the poles of the boundary correlators correspond to the spectrum of the quasinormal modes (QNMs) in the dual bulk. One important property on which some of our analysis depended was the thermal product formula (or the thermal product hypothesis) of Ref.~\cite{Dodelson:2023vrw}. Owing to the special asymptotic properties of such meromorphic correlators, we were able to rigorously establish a number of statements on the distributions and densities of zeroes and poles of the large-$N$ two-point correlators. Moreover, we discussed in detail the spectral duality relation originally derived in our Ref.~\cite{Grozdanov:2024wgo} and further expanded upon in Ref.~\cite{Grozdanov:2025ner}. We showed that this duality relation between pairs of spectra can be extended and applied to theories and correlators beyond those of the S-dual 3$d$ CFTs, in particular, to correlators in theories in any number of dimensions that are related by the Legendre transform. A particularly interesting example are the UV and IR CFTs related by the double-trace RG flows. In the process, we studied a number of examples that demonstrate the validity of our assumptions on the properties of correlators and the spectral duality relation. Furthermore, we developed a simple numerical algorithm that could be implemented to directly compute the spectrum of poles (and zeroes) from the spectrum of another, dual retarded correlator. An example of such a construction allowed us to compute the poles and zeroes of a retarded scalar primary two-point correlator in the IR CFT directly from only the spectrum of poles of the scalar correlator in the UV CFT.  

While most of the assumptions made in Section~\ref{subsec:analytic_properties} that underlie all of this paper are expected to hold quite generally, the thermal product hypothesis might need to be adapted in situations with non-vanishing equilibrium expectation values of the operators. In holographic systems, this usually implies a non-trivial coupling between the equations of motion of various fields resulting in a system of differential equations more complicated than a single bulk wave equation. This scenario was recently investigated in Ref.~\cite{Bhattacharya:2025vyi}, where a modification of the thermal product formula was proposed. It would be interesting to investigate a similar modification to the spectral duality relation. Another important generalisation of the results discussed here pertains to the study of large-$N$ correlators that might exhibit a branch cut, such as the one computed in Ref.~\cite{Romatschke:2019ybu}. Of course, a more thorough and rigorous analysis of correlators functions and spectra with poles, as well as branch cuts also remains an important open problem in QFTs and CFTs beyond the large-$N$ limit, as does the search for explicit analytically tractable examples of such correlators. 

We note that another promising direction for further exploration would be to use more detailed OPE data to compute corrections to the Christmas tree asymptotics of both the UV and IR spectra (see Ref.~\cite{Dodelson:2023vrw}), and use that in conjunction with the spectral duality relation to constrain low-energy physics. Furthermore, the formalism presented in this paper should also be used alongside the thermal bootstrap techniques discussed in Refs.~\cite{Caron-Huot:2009ypo,Iliesiu:2018fao,Petkou:2018ynm,Karlsson:2019qfi,Karlsson:2022osn,Marchetto:2023xap,Ceplak:2024bja,Buric:2025anb,Barrat:2025nvu,Buric:2025fye} to potentially even further constrain the spectra of large-$N$ thermal correlators. Finally, the discussion of Section~\ref{sec:3d} could also be generalised, for example, to higher-dimensional large-$N$ systems with Maxwell dynamics \cite{Grozdanov:2017kyl,Hofman:2017vwr,Grozdanov:2018fic,DeWolfe:2020uzb,Frangi:2024enh} and to operators of other spins. We defer all such explorations to future works.

\appendix
\section{Real-time formalism and relevant applications}
\label{app:RT_CT}
In this appendix, we present a few relevant details that are required to specify retarded correlation functions, as well as discuss the real-time formulation of the double-trace deformed RG flows.

\subsection{Contact terms}
\label{subapp:contact}
The discussion about the double-trace deformed spectra depends critically on the precise definition of the retarded correlators. Various definitions of correlators can differ by contact terms, which can arise when the correlators are defined variationally, and the action depends non-linearly on the source \cite{Romatschke_2009}. This usually happens where the source is turned on by gauging the derivatives, for example, when the theory has an external U$(1)$ gauge field turned on or is placed on a curved background. We work on a (Schwinger-Keldysh) closed-time-path contour running from $t=-\infty$ to $t=\infty$ (branch $1)$, then back to $t=-\infty$ (branch $2$) and around the thermal circle (along a branch with imaginary time). We can write
\begin{equation}
    e^{\Gamma[J_1,J_2]}=\int \CD \varphi_E \CD \varphi_1 \CD \varphi_2 \, e^{-S_E[\varphi_E]+i S\qty[\varphi_1,J_1]-i S\qty[\varphi_2,J_2]}, \label{app:CTP}
\end{equation}
where the action $S$ is a functional of the dynamical field $\varphi$, as well as the source $J$, so that
\begin{equation}
    \CO(t)=\left.\frac{\delta S[\varphi,J]}{\delta J(t)}\right|_{J=0}.
\end{equation}
Note that we do not demand that the action is linear in the source. $S_E$ is the Euclidean action, and the path integral is path-ordered with respect to the integration contour. The causal (retarded) response to an external source (see Eq.~\eqref{eq:green_response}) is
\begin{equation}
    \expval{\CO(t)}_J=-\frac{i}{2}\left.\qty[\frac{\delta\Gamma[J_1,J_2]}{\delta J_1(t)}-\frac{\delta\Gamma[J_1,J_2]}{\delta J_2(t)}]\right|_{\substack{J_1=J_2=J}}, \label{app:causal_response}
\end{equation}
which, after the sources had been identified, gives
\begin{equation}
    \expval{\CO(t)}_J=\frac{1}{2}\int \CD \varphi_E \CD \varphi_1 \CD \varphi_2 \qty[\frac{\delta S[\varphi_1,J]}{\delta J(t)}+\frac{\delta S[\varphi_2,J]}{\delta J(t)}]e^{-S_E[\varphi_E]+i S\qty[\varphi_1,J]-i S\qty[\varphi_2,J]}.
\end{equation}
Here, we used the fact that $\Gamma[J,J]=0$ due to unitarity. Defining the retarded correlator as
\begin{equation}
    G(t)=\left.\frac{\delta \expval{O(t)}_J}{\delta J(0)}\right|_{J=0}, \label{app:retard}
\end{equation}
we get
\begin{equation}
    G(t,\vb{x})=i \Theta(t) \expval{[\CO(t,\vb{x}),\CO(0,0)]}+\expval{\frac{\delta^2 S[J]}{\delta J(0,0) \delta J(t,\vb{x})}}_{J=0}. \label{app:contact}
\end{equation} 
The first term in Eq.~\eqref{app:contact} is simply the canonical definition of the retarded correlator. The second term arises if the action is non-linear in the source. For local actions, it is a contact term, i.e., a Dirac delta function or its $n$-th derivative. In Fourier space, it therefore corresponds to a polynomial of $\omega$ and $\vb k$, assuming that the action is of finite order in derivatives.

\subsection{Double-trace deformations in the Lorentzian signature}
\label{subapp:trace}
Now we consider double-trace deformations in the real-time formalism. For simplicity, we assume that the action is now linear in the source, i.e., that  $S[\varphi,J]= S_0[\varphi]+\int dt \, \CO(t) J(t)$, where we suppressed the spatial dependence. We consider a double-trace deformed action
\begin{equation}
    S_f[\varphi,J]=S_0[\varphi]+\int dt\qty[\CO(t) J(t)-\frac{f}{2}\CO(t)^2].
\end{equation}
The corresponding generating functional is defined as
\begin{equation}
    e^{\Gamma[J_1,J_2;f]}=\int \CD \varphi_E \CD \varphi_1 \CD \varphi_2 \, e^{-S_E[\varphi_E]+i S_f\qty[\varphi_1,J_1]-i S_f\qty[\varphi_2,J_2]}.
\end{equation}
Note that we did not deform the Euclidean part of the path integral as that turns out to be inconsequential for the computation of two-point functions (see e.g.~Ref.~\cite{landsman87}). Recalling Eq.~\eqref{eq:Hubbard-Stratonovich}, we can write the partition function of the deformed theory as a path integral over the source,
\begin{equation}
    \expval{e^{i\int dt\qty[ \CO(t) J(t)-\frac{f}{2}\CO(t)^2]}}=\int \CD \widetilde{J} \, e^{i\int dt\qty[ \CO(t) \widetilde J(t)+\frac{1}{2f}(J(t)-\widetilde J(t))^2]}.
\end{equation}
Applying this to both sides of Eq.~\eqref{app:CTP}, we get
\begin{equation}
    e^{\Gamma[J_1,J_2;f]}=\int \CD \widetilde J_1 \CD \widetilde J_2 \, e^{\Gamma[\widetilde J_1,\widetilde J_2,0]+\frac{i}{2f}\int dt\qty[(J_i(t)-\widetilde J_i(t))\CF_{ij}(J_j(t)-\widetilde J_j(t))]},\label{app:SK_dbltrace}
\end{equation}
where
\begin{equation}
    \Gamma[J_1,J_2;f]=\frac{i}{2}\int dt_1 dt_2 \, J_i(t_1) \CG_{ij}(t_1-t_2;f)J_j(t_2),
\end{equation}
and
\begin{equation}
    \CF=\mqty(1& 0 \\ 0 & -1).
\end{equation}
We can perform the path integral from Eq.~\eqref{app:SK_dbltrace} in the saddle point approximation, and arrive at the matrix equation
\begin{equation}
    \CG(\omega;f)=\frac{1}{f}\qty[\mathbb{I}-(f \CG(\omega;0) \CF+\mathbb{I})^{-1}]\CF.
\end{equation}
Finally, we can then compute the retarded correlator as defined through Eqs.~\eqref{app:causal_response} and \eqref{app:retard}, and arrive at
\begin{equation}
    G(\omega;f)=\frac{G_0(\omega)}{1+f G_0(\omega)},
\end{equation}
which is the equation for the double-trace deformed retarded correlator that we have been using throughout this work. Note that the same equation trivially holds for the advanced correlator, but not for other Lorentzian two-point functions.

\section{Details of the relevant aspects of complex analysis}
\label{app:complex-analysis}
In this appendix, we elaborate on the details of complex analysis that underlies the statements made in the main part of the paper, and present the relevant proofs. We first focus on the asymptotic properties of $\rho(\omega)$ and $G(\omega)$, their decomposition, and the respective implications for the distribution of zeroes and poles. We then turn to the implications of the thermal product formula, concluding with the derivation of the spectral duality relation. For clarity, we repeat the main assumptions and definitions.

The hermiticity property \eqref{eq:hermiticity0} and the analyticity in the upper complex half-plane \eqref{assumption:retarded} of the retarded correlators are assumed throughout, as well as the absence of any higher-order poles or zeroes in the correlation functions.

\subsection{A theorem on partial fractions decomposition}
\label{subapp:partial}
We first cite a theorem on the partial fractions decomposition of a class of meromorphic functions (see e.g.~Refs.~\cite{smirnov, saksAnalyticFunctions}).
\begin{theorem}
\label{theorem:partial_fraction}
    Let $R_N$ be a sequence of positive real numbers such that
    \begin{equation}
        R_N<R_{N+1}, \quad \lim_{N\rightarrow\infty}R_N =\infty. \label{app:RN}
    \end{equation}
    Let $f(\omega)$ be a meromorphic function, with the set of poles denoted by $\qty{\omega_n}$, and corresponding residues denoted by $\qty{r_n}$. All the poles are assumed to be simple. Furthermore, let $M_N$ denote the maximal modulus of $f(\omega)$ on concentric circles
    \begin{equation}
        M_N = \max_{\abs{\omega}=R_N}\abs{f(\omega)}.
    \end{equation}
    Suppose that $f(\omega)$ is such that, for some sequence $R_N$, we have
    \begin{equation}
        \lim_{N\rightarrow \infty} M_N=0. \label{app:th:decay_assumption}
    \end{equation}
    Then, $f(\omega)$ admits a partial fractions expansion 
    \begin{equation}            f(\omega)=\lim_{N\rightarrow\infty}\sum_{\abs{\omega_n} < R_N} \frac{r_n}{\omega-\omega_n}. \label{app:th:partial}
    \end{equation} 
    For any $R>0$, the convergence in the region $\abs{\omega}\leq R$ is uniform.
\end{theorem}

\begin{proof}
Consider the positively oriented integral
    \begin{equation}
        \CI_N(\omega) = \frac{1}{2\pi i}\oint\limits_{\abs{z}=R_N}\frac{f(z)}{z-\omega}dz.
    \end{equation}
Let $\omega$ be such that $\abs{\omega}\leq R$ for some $R>0$, and $N$ large enough so that $R_N>R$. The integral can be evaluated using the residue theorem, which gives
\begin{equation}
    \CI_N(\omega)=f(\omega)-\sum_{\abs{\omega_n}<R_N}\frac{r_n}{\omega-\omega_n}. \label{app:proof:IN}
\end{equation}
On the other hand, $\CI_N(\omega)$ can be estimated:
\begin{equation}
    \abs{\CI_N(\omega)}\leq \frac{1}{2\pi} \oint\limits_{\abs{z}=R_N}\abs{\frac{f(z)}{z-\omega}}\abs{dz}\leq \frac{R_N M_N}{R_N-R} .\label{app:proof:estimate}
\end{equation}
Assumption \eqref{app:th:decay_assumption} then guarantees that the right-hand side of Eq.~\eqref{app:proof:estimate} can be made arbitrarily small, i.e., that
\begin{equation}
    \lim_{N\rightarrow\infty}\abs{\CI_N(\omega)}=0.
\end{equation}
The theorem is proven by noticing that the left-hand side of Eq.~\eqref{app:proof:IN} goes to zero in the $N\rightarrow \infty$ limit. Since the estimate on the right-hand side of Eq.~\eqref{app:proof:estimate} is independent of $\omega$, the partial fractions expansion is uniformly convergent for all $\omega$ with $\abs{\omega}\leq R$, where $R$ can be arbitrarily large.
\end{proof}

\begin{corollary}
\label{corollary:partial}
Assumption \ref{app:th:decay_assumption} can be relaxed to allow for a broader class of meromorphic functions to be considered. Specifically, suppose that $f(\omega)$ is such that, for some sequence $R_N$, and some positive integer $p$, we have
\begin{equation}
    \lim_{N\rightarrow\infty} \frac{M_N}{R_N^p}=0.
\end{equation}
Furthermore, suppose that $f(0)\neq \infty$. An auxiliary meromorphic function $f^{\qty[p]}(\omega)$ can then be defined as
\begin{equation}
    f^{\qty[p]}(\omega)=\omega^{-p}\qty[f(\omega)-\sum_{m=0}^{p-1} f_m \omega^{m}],\label{app:def:palt}
\end{equation}
where $f_m$ are the Taylor series coefficients of $f(\omega)$ around $\omega=0$,
\begin{equation}
    f_m=\left.\frac{\partial^m_\omega f(\omega)} {m!}\right|_{\omega=0}. 
\end{equation}
The $\omega^{-p}$ factor in Eq.~\eqref{app:def:palt} ensures that the maximal modulus of $f^{\qty[p]}$ on concentric circles decays to zero as $N\rightarrow\infty$, while the second term in Eq.~\eqref{app:def:palt} ensures that there is no singularity at $\omega=0$. The function $f^{\qty[p]}(\omega)$ then obeys all of the assumptions of Theorem \ref{theorem:partial_fraction}, and admits a partial fractions decomposition in the sense of Eq.~\eqref{app:th:partial}. Finally, the convergent partial fractions expansion of $f(\omega)$ then takes the explicit form
\begin{equation}
    f(\omega)=\sum_{m=0}^{p-1}f_m \omega^m + \omega^p\lim_{N\rightarrow\infty}\sum_{\abs{\omega_n}< R_N} \frac{r_n}{\omega_n^p}\frac{1}{\omega-\omega_n}.
\end{equation}
\end{corollary}

\subsection{Assumption on the power-law boundedness of $\rho(\omega)$}
\label{subapp:boundedness}
In real time, the spectral function takes the form $\rho(t)=\expval{[\CO(t),\CO(0)]}/2$. Aside from the possible singular behaviour at $t=0$, it is expected to be a finite, well-behaved function of $t$. Consequently, $\rho(\omega)$ must be such that it admits a well-behaved Fourier transform. A sufficient condition that guarantees such behaviour is captured in the assumption on the power-law boundedness of $\rho(\omega)$, namely, in Eq.~\eqref{assumption:rho_bounded}.\footnote{Note that there exist meromorphic functions which admit a Fourier transform for $t\neq0$, which, in some directions of the complex plane, grow faster than a power law but slower than an exponential. An example of such function is $f(\omega)=\cos \sqrt{\omega} \cos\sqrt{-\omega}/\cosh(\omega-1)\cosh(\omega+1)$, the Fourier transform of which can be computed using residues, even though $\abs{f(\omega\rightarrow i\infty)}\sim e^{\sqrt{\abs{\omega}}}$. We do not consider the possibility of such spectral functions in this paper.} It states that there exists a sequence $R_N$ (in the sense of Eq.~\eqref{app:RN}), so that
\begin{equation}
    \max_{\abs{\omega}=R_N} \abs{\rho(\omega)}< A R_N^{2\Delta-d}, \label{app:rho_powerlaw}
\end{equation}
for some large enough constant $A$. The $2\Delta-d$ exponent is motivated by the large-$\omega$ behaviour of $\rho(\omega)$ (see Eq.~\eqref{assumption:OPE}). This assumption guarantees that $\rho(\omega)$ admits a partial fractions expansion in terms of Theorem \ref{theorem:partial_fraction} and thus a Fourier transform that is well-behaved for $t\neq 0$. More specifically, if $2\Delta-d<0$, Theorem \ref{theorem:partial_fraction} can be used directly, and we have
\begin{subequations}
\label{app:partial-frac-rho}
\begin{equation}
    \rho(\omega)=\lim_{N\rightarrow\infty} \sum_{\abs{\omega^+_n}<R} \frac{r_n}{\omega-\omega_n^+}+\frac{r_n}{\omega+\omega_n^+},
\end{equation}
whereas for $2\Delta-d \geq 0$, Corollary \ref{corollary:partial} can be used to write
\begin{equation}
    \rho(\omega)=\sum_{m=0}^{p-1} \rho_m \omega^m+\omega^p \lim_{N\rightarrow\infty} \sum_{\abs{\omega^+_n}<R} \frac{r_n}{(\omega_n^+)^p}\qty[\frac{1}{\omega-\omega^+_n}+\frac{(-1)^p}{\omega+\omega_n^+}],
\end{equation}
\end{subequations}
where $p>2\Delta-d$. Here, we have used the fact that $\rho(\omega)$ is an odd function, with poles at $\qty{\pm \omega_n^+}$.

\subsection{Assumption on the power-law asymptotics of $G(\omega)$}
\label{subapp:powerlaw}
We now consider to the retarded correlator and the assumption on its power-law asymptotics. This assumption is motivated by the discussion in Section~\ref{subsec:UV-IR-CFT}, where $-1/G(\omega)$ can be interpreted as a retarded correlator in its own right, and is thus expected to have a sensible real-time representation. The expectation is therefore that $G(\omega)$ exhibits power-law behaviour in all directions of the complex plane that admit an asymptotic expansion. This gives $\partial_\omega \ln G(\omega) \sim \omega^{-1}$. Instead of relying strictly on the asymptotic expansion along the rays in the complex plane, we assume that there exists a sequence $R_N$ (which may differ from the one used in Eq.~\eqref{app:rho_powerlaw}), for which\footnote{Note that we do not explicitly assume that $G(\omega)$ itself is bounded on concentric circles in the sense of Eq.~\eqref{app:rho_powerlaw}, which would guarantee a well-behaved Fourier transform. It turns out that such an assumption would be largely redundant in our discussion.}
\begin{equation}
    \max_{\abs{\omega}=R_N}\abs{\partial_\omega \ln G(\omega)}\leq \frac{C}{R_N}. \label{app:G-powerlaw}
\end{equation}
Since $G(\omega)$ is meromorphic with only simple poles $\qty{\omega_n^+}$ and zeroes $\qty{\omega_n^-}$, we must have
\begin{subequations}
\begin{align}
\text{Res}_{\omega_n^+}\partial_\omega \ln G(\omega)&=-1,\\
\text{Res}_{\omega_n^-}\partial_\omega \ln G(\omega)&=+1.
\end{align}
\end{subequations}
We can now use assumption \eqref{app:G-powerlaw} and Theorem \ref{theorem:partial_fraction} to write
\begin{equation}
    \partial_\omega \ln G(\omega)=\lim_{N\rightarrow \infty} \qty[\sum_{\abs{\omega_n^-}<R_N}\frac{1}{\omega-\omega_n^-}-\sum_{\abs{\omega_n^+}<R_N}\frac{1}{\omega-\omega_n^+}]. \label{app:G_dlog}
\end{equation}
The limit above is uniformly convergent on any closed, centred disc. Consequently, Eq.~\eqref{app:G_dlog} can be integrated and exponentiated to give\footnote{In contrast with the more usual Weierstrass decomposition of a meromorphic function, the decomposition presented here has only a single parameter that is not directly determined by poles and zeroes, namely, $G(0)$.}
\begin{equation}
\label{app:G-product}
G(\omega)=G(0)\lim_{N\rightarrow\infty}\frac{\pi_N^-(\omega)}{\pi_N^+(\omega)},
\end{equation}
with
\begin{equation}
    \pi^\pm_N(\omega)=\prod_{\abs{\omega_n^\pm}<R_N}\qty(1-\frac{\omega}{\omega_n^\pm}).
\end{equation}

To discuss the implications of the assumption on the power-law asymptotics of $G(\omega)$, we remind the reader that $n_+(R)$ ($n_-(R)$) counts the numbers of poles (zeroes) of $G(\omega)$ with $\abs{\omega_n^+}<R$ ($\abs{\omega_n^-}<R$). The argument principle tells us that
\begin{equation}
    n_-(R)-n_+(R)=\frac{1}{2\pi i}\oint\limits_{\abs{\omega}=R}\partial_\omega \ln G(\omega) d\omega,
\end{equation}
where the integral is positively oriented. Using the triangle-type inequality, assumption \eqref{app:G-powerlaw} then allows us to estimate
\begin{equation}
    \abs{n_-(R_N)-n_+(R_N)}\leq\frac{1}{2\pi}\oint\limits_{\abs{\omega}=R_N}\abs{\partial_\omega \ln G(\omega)} \abs{d\omega}\leq C.
\end{equation}
This means that there exists a sequence $R_N$ for which the difference between the number of zeroes and poles in the radius-$R_N$ disc is bounded by some constant $C$. Since there is an infinite number of such points, we can write
\begin{equation}
    \liminf_{R\rightarrow\infty}\abs{n_-(R)-n_+(R)}\leq C.
\end{equation}
This is the precise sense in which the radial densities of poles and zeroes must essentially be the same. Besides this fact, the poles and zeroes they must also lie close together. More specifically, the sum
\begin{equation}
    \lim_{N\rightarrow \infty} \abs{\sum_{\abs{\omega_n^-}<R_N}\frac{1}{\omega_n^-}-\sum_{\abs{\omega_n^+}<R_N}\frac{1}{\omega_n^+}}<\infty
\end{equation}
is convergent. The statement follows from the evaluation of Eq.~\eqref{app:G_dlog} at $\omega=0$. 

\subsection{Assumption on the large-$\omega$ behaviour of $\rho(\omega)$}
\label{subapp:OPE}
Generic properties of a CFT, namely, the OPE, tell us that $\rho(\omega)$ admits an asymptotic expansion with the leading term scaling as \cite{Dodelson:2023vrw,Caron-Huot:2009ypo}
\begin{equation}
    \omega\rightarrow\infty^+: \quad \rho(\omega) \sim \omega^{2\Delta-d}. \label{app:rhoOPE}
\end{equation}
Quantitatively, this provides constraints on the Christmas tree spectrum (see Eqs.~\eqref{eq:christmas_constraint} and \eqref{eq:xmas_constraint}), which will be discussed in Appendix~\ref{subapp:xmas-OPE}. Here, we use Eq.~\eqref{app:rhoOPE} to justify the partial fractions expansion of $G(\omega)$ (see Eq.~\eqref{eq:G_partial_decomp}). Given that $\rho(\omega)$ is proportional to the odd part of $G(\omega)$, with the poles present only in the lower half-plane, we can use Eq.~\eqref{app:partial-frac-rho} to write
\begin{subequations}
\label{app:G-partialfrac}
\begin{align}
    2\Delta-d &< p: & G(\omega)&=g(\omega^2)+i \sum_{m=0}^{p-1}\rho_m \omega^m +2i\omega^p \lim_{N\rightarrow\infty}\sum_{\abs{\omega_n^+}<R_N} \frac{r_n}{(\omega^+_n)^p}\frac{1}{\omega-\omega^+_n},\\
    2\Delta-d&<0: & G(\omega)&=g(\omega^2)+2i \lim_{N\rightarrow\infty}\sum_{\abs{\omega_n^+}<R_N}  \frac{r_n}{\omega-\omega_n^+}.
\end{align}    
\end{subequations}
The even function $g(\omega^2)$ is set to be a polynomial by the assumption on the asymptotic power-law behaviour \eqref{app:G-powerlaw}. The convergence of the sums in Eq.~\eqref{app:G-partialfrac} is not guaranteed by the convergence of the sums in the partial fractions decomposition of $\rho(\omega)$. While an extra assumption on the power-law boundedness of $G(\omega)$ would amend that trivially, we instead choose to rely on Eq.~\eqref{app:rhoOPE} to show that the decomposition prescription of Eq.~\eqref{app:G-partialfrac} is indeed convergent.

To proceed, we introduce $\rho^{[p]}$ as
\begin{equation}
    \rho^{[p]} =
    \begin{cases}
        \rho(\omega), & p=0,\\
        \omega^{-p}\qty[\rho(\omega)-\sum\limits_{m=0}^{p-1}\rho_m \omega^m], & p>0,
    \end{cases}
\end{equation}
with $p>2\Delta-d$. We also introduce
\begin{equation}
    M_N = \max_{\abs{\omega}=R_N}\abs{\rho^{[p]}(\omega)}.
\end{equation}
By construction of $\rho^{[p]}$ and assumption \eqref{app:rho_powerlaw}, we have
\begin{equation}
    \lim_{N\rightarrow\infty}M_N=0,
\end{equation}
and, from assumption \eqref{app:rhoOPE},
\begin{equation}
    \lim_{\omega\rightarrow \infty^+}\rho^{[p]}(\omega)=0.
\end{equation}
Since, by assumption, there are no poles on the real axis, $\rho^{[p]}(\omega)$ is bounded along the entire real axis. Let $\CC_N$ be a positively oriented contour in the complex plane that is the boundary of a radius-$R$ half-disc in the lower complex half-plane, i.e.,
\begin{equation}
    \oint_{\CC_N} = \int_{\CC_N^1} + \int_{\CC_N^2}, \label{app:contour}
\end{equation}
where (see Fig.~\ref{fig:half-disc})
\begin{subequations}
 \begin{align}
    \int_{\CC_N^1} &= \int\limits_{\substack{\Im\omega =-\epsilon\\\abs{\omega}\leq R_N}}, \\
    \int_{\CC_N^2} &= \int\limits_{\substack{\abs{\omega}=R_N\\\Im\omega<0}}.
\end{align}   
\end{subequations}
\begin{figure}
	\centering
	\begin{tikzpicture}[scale=1]
		\node[rectangle,minimum width=.5cm,minimum height=.5cm,anchor=north east,draw] (lab) at (7,2) {$\omega$};
        \node[inner sep=0pt,anchor=south] at (6.5,.7+0.2) {$R_N$};
        \draw[very thin,gray] 
			(1,.7) -- (7,.7)
			(4,-2.5) -- (4,2);
        \draw
            (1.5,.5) arc [radius=2.5, start angle=-180, end angle= 0] node[currarrow,pos=0.66,xscale=1,sloped,scale=1.25] {};
        \draw
            (1.5,.5) -- (6.5,.5) node[currarrow,pos=0.25,xscale=-1,sloped,scale=1.25] {};
        \node[inner sep=0pt,anchor=north east] at ($(4,.5)+ (.3,.3)+2.5*cos(225)*(1,0)+2.5*sin(255)*(0,1)$) {$\mathcal C_N^2$};
        \node[inner sep=0pt,anchor=north] at (5.25,.35) {$\mathcal C_N^1$};
        \end{tikzpicture}
        \caption{The positively-oriented contour $\CC_N$ used in Eq.~\eqref{app:contour}. \label{fig:half-disc}}
\end{figure}
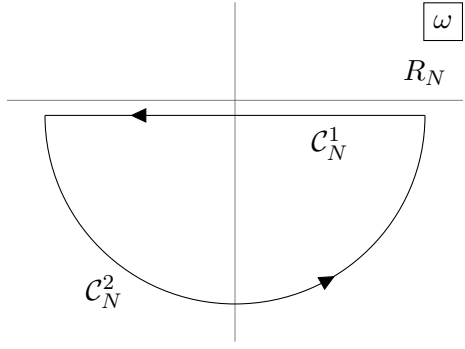
We consider the integral
\begin{equation}
    \CI_N(\omega)=\frac{1}{2\pi i}\lim_{\epsilon\rightarrow0}\oint_{\CC_N}\frac{\rho^{\qty[p]}(z)}{z-\omega}dz,
\end{equation}
where $\omega$ is real, with $-R_N\leq \omega \leq R_N$. On the one hand, this integral can be evaluated by using the residue theorem. The only poles of the integrand that the $\CC_N$ winds around are the poles in the lower half-plane, which gives
\begin{equation}
    \CI_N(\omega) = -\sum_{\abs{\omega_n^+}<R_N}\frac{r_n}{(\omega^+_n)^p}\frac{1}{\omega-\omega_n^+}. \label{app:in_convergence}
\end{equation}
This is precisely the sum that appears in the partial fractions decomposition of the retarded correlator \eqref{app:G-partialfrac}. We will now show that the limit $N\rightarrow\infty$ in Eq.~\eqref{app:in_convergence} converges to a finite value for all real $\omega$. We introduce
\begin{equation}
    \CI_N=\CI_N^1+\CI_N^2,
\end{equation}
where
\begin{subequations}
\begin{align}
\CI^1_N &= \frac{1}{2\pi i}\lim_{\epsilon\rightarrow0}\int_{\CC^1_N}\frac{\rho^{\qty[p]}(z)}{z-\omega}dz,  \\
\CI^2_N &= \frac{1}{2\pi i}\lim_{\epsilon\rightarrow0}\int_{\CC^2_N}\frac{\rho^{\qty[p]}(z)}{z-\omega}dz.
\end{align}
\end{subequations}
The first part of the contour amounts to evaluating
\begin{equation}
    \CI_N^1(\omega) =- \frac{1}{2\pi i} \lim_{\epsilon\rightarrow 0}\int_{-R_N}^{R_N} \frac{\rho^{\qty[p]}(z)}{z-\omega-i\epsilon}dz.
\end{equation}
Using the Sokhotski-Plemelj theorem, we can evaluate
\begin{equation}
    \CI_N^1(\omega)=-\frac{1}{2}\rho^{[p]}(\omega)-\frac{1}{2\pi i}\CP\int_{-R_N}^{R_N}\frac{\rho^{\qty[p]}(z)}{z-\omega}dz,
\end{equation}
where $\CP$ stands for Cauchy principal value. The first term is finite for all the $\omega$ considered, and is independent of $N$. The second term can be also shown to be finite for all $\omega$ by splitting it into three parts and using the fact that $\rho^{[p]}(\omega)$ is a smooth, bounded function that vanishes at $\omega\rightarrow \infty$. Turning now to $\CI_N^2$, we proceed as in the proof of Theorem \ref{theorem:partial_fraction}. Due to $\rho(\omega)$ being odd, we have the identity
\begin{equation}
    M_N=\max_{\omega\in \CC_N^2}\abs{\rho^{[p]}(\omega)}.
\end{equation}
We can now estimate
\begin{equation}
    \abs{\CI_N^2(\omega)}\leq \frac{1}{2\pi}\int_{\CC_N^2}\abs{\frac{\rho^{\qty[p]}(z)}{z-\omega}}\abs{dz}\leq \frac{1}{2}\frac{R_N M_N}{R_N-\abs{\omega}}.
\end{equation}
We therefore have
\begin{equation}
    \lim_{N\rightarrow\infty}\CI_N^2(\omega)=0.
\end{equation}
We conclude that the $N\rightarrow\infty$ limit of $\CI_N(\omega)$ is finite for all real $\omega$. Therefore, the sum \eqref{app:G-partialfrac} converges, and an explicit partial fractions expansion of $G(\omega)$ exists.\footnote{Note that, for real $\omega$, this gives the Kramers-Kronig relation
\begin{equation}
    \Re G(\omega)= g(\omega^2)+\frac{\omega^p}{\pi}\CP \int_{-\infty}^\infty \frac{\rho^{[p]}(z)}{z-\omega}dz.
\end{equation}
In this form, the integral is appropriately regularised to accommodate the possibly divergent spectral function. This is in contrast with the usual form of the relation, where $p=0$ and the integral may be divergent. In our case, the scheme dependence of the divergent integral is absorbed into the free coefficients of $g(\omega^2)$. This identity is not limited to meromorphic correlators.
} Evaluating $\CI_N(0)$ gives the convergence of the sum \eqref{eq:res_convergence}.
\subsection{A theorem on even entire functions of order one}
\label{subapp:entire}
In this appendix, we cite some standard definitions and results from the theory of entire functions, and present a theorem which is applicable directly to the thermal product hypothesis. For details and proofs, we refer the reader to Refs.~\cite{levin80,levin96,Holland73}.

A function $f(\omega)$ is entire if it is holomorphic on the entire complex plane. As such, it is characterised by a set of isolated zeroes $\qty{\omega_n}$, as well as an absence of poles or branch points. An entire function $f(\omega)$ is of finite order if $\ln\abs{f(\omega)}$  is asymptotically bounded by $\abs{\omega}^{\rho+\epsilon}$ for some $\rho$ and any $\epsilon > 0$. More specifically, the order of $f(\omega)$, denoted by $\rho_f$, is defined as
\begin{equation}
\rho_f=\limsup_{R\rightarrow\infty}\frac{\ln\ln M_f(R)}{\ln R},
\end{equation}
where
\begin{equation}
M_f(R)=\max_{\abs{\omega}=R}\abs{f(\omega)}.
\end{equation}
The asymptotic behaviour of an entire function is intimately related to the distribution of its zeroes, which we assume to be simple zeroes. To quantify this, we consider the sum
\begin{equation}
    \sum_n \frac{1}{\abs{\omega_n}^\mu} < \infty,
\end{equation}
for some large enough, non-negative $\mu$. The \emph{convergence exponent} is the lower bound of all such $\mu$, i.e., $\lambda=\inf\mu$, while the \emph{genus} is the smallest non-negative integer $p$ such that the sum converges for $\mu=p+1$. The counting function $n(R)$, which counts the number of zeroes with $\abs{\omega_n}<R$, is related to the convergence exponent via
\begin{equation}
    \limsup_{R\rightarrow\infty}\frac{\ln n(R)}{\ln R}=\lambda, \label{app:th:limsup}
\end{equation}
and a theorem by Hadamard tells us that
\begin{equation}
    \lambda \leq \rho_f. \label{app:th:Hadamard}
\end{equation}
The set $\qty{\omega_n}$ then uniquely defines an absolutely convergent \emph{canonical} Hadamard product
\begin{equation}
    g(\omega)=\prod_n e^{P_p(\omega)} \qty(1-\frac{\omega}{\omega_n}), \label{app:th:canonical}
\end{equation}
where
\begin{equation}
    P_p(\omega)=\sum_{j=1}^p \frac{\omega}{j\omega_n}.
\end{equation}
We note that a product $\prod_n (1+a_n)$ is absolutely convergent if $\sum_n \abs{a_n}<\infty$. A theorem by Borel then states that the order of the canonical Hadamard product $\rho_g$ coincides with the corresponding exponent of convergence, i.e.,
\begin{equation}
    \rho_g=\lambda. \label{app:th:Borel}
\end{equation}
Finally, the Hadamard factorisation theorem states that any entire function of finite order may be decomposed in the following way,
\begin{equation}
    f(\omega)=g(\omega) e^{q(\omega)}, \label{app:th:HadamardDecomp}
\end{equation}
where $q(\omega)$ is a polynomial of degree no greater than $\rho_f$. We can now prove a theorem on even entire functions of order one.
\begin{theorem}
\label{app:entiretheorem}
Let $f(\omega)$ be an even entire function of order $\rho_f=1$, with only simple zeroes, denoted by $\qty{\omega_n,-\omega_n}$, $\omega_n\neq 0$, and let $n(R)$ count all $\omega_n$ with $\abs{\omega_n}<R$. The function $f(\omega)$ can then be expressed as
    \begin{equation}
        f(\omega)=f(0)\prod_n \qty(1-\frac{\omega^2}{\omega_n^2}),
    \end{equation}
    with the counting function obeying
    \begin{equation}
        \limsup_{R\rightarrow\infty}\frac{\ln n(R)}{\ln R}=1. \label{app:th:limsupT}
    \end{equation}
\end{theorem}
\begin{proof}
Due to Hadamard's theorem (cf.~Eq.~\eqref{app:th:Hadamard}) we have $p\leq\lambda\leq 1$, which sets the genus to either $p=0$ or $p=1$. The canonical Hadamard products take the form (see Eq.~\eqref{app:th:canonical})
\begin{align}
    p&=0: \quad g(\omega)=\prod_n \qty(1-\frac{\omega}{\omega_n})\prod_n \qty(1+\frac{\omega}{\omega_n})=\prod_n \qty(1-\frac{\omega^2}{\omega_n^2}),\\
    p&=1: \quad g(\omega)=\prod_n e^{\omega/\omega_n}\qty(1-\frac{\omega}{\omega_n})\prod_n e^{-\omega/\omega_n} \qty(1+\frac{\omega}{\omega_n})=\prod_n \qty(1-\frac{\omega^2}{\omega_n^2}).
\end{align}
Since the function is even, the canonical products themselves are even, and take the same form irrespective of the genus. By the Hadamard factorisation theorem \eqref{app:th:HadamardDecomp}, we have
\begin{equation}
    f(\omega)=f(0)g(\omega) e^{a\omega}.
\end{equation}
Since both $f(\omega)$ and $g(\omega)$ are even, we must set $a=0$. The function $f(\omega)$ is then just a rescaling of the canonical product $g(\omega)$, which sets the convergence exponent to $\lambda=1$ due to Borel's theorem (cf.~Eq.~\eqref{app:th:Borel}). Following Eq.~\eqref{app:th:limsup}, we obtain Eq.~\eqref{app:th:limsupT}. Note that the fact that $n(R)$ counts the zeroes without their respective mirrored pairs does not affect Eq.~\eqref{app:th:limsupT}.

\end{proof}

\subsection{The thermal product formula}
\label{subapp:TPF}
In this appendix, we recap the main result of Ref.~\cite{Dodelson:2023vrw} and endow it with a constraint on the asymptotic distribution of poles. The thermal product hypothesis states that
\begin{equation}
    \frac{\sinh\frac{\beta\omega}{2}}{\rho(\omega)} \text{ is an entire function of order one}. \label{app:hypothesis}
\end{equation}
Consequently, we can apply Theorem \ref{app:entiretheorem} to write
\begin{equation}
    \frac{\sinh\frac{\beta\omega}{2}}{\rho(\omega)} = \frac{1}{\mu} \prod_n \qty[1-\qty(\frac{\omega}{\omega_n^+})^2], \label{app:tpf}
\end{equation}
with the counting function $n_+(R)$ obeying
\begin{equation}
    \limsup_{R\rightarrow\infty}\frac{\ln n_+(R)}{\ln R}=1. \label{app:n_limsup}
\end{equation}
Eq.~\eqref{app:tpf} allows us to compute the residues of $\rho(\omega)$,
\begin{equation}
    r^+_n=-\mu\frac{\omega^+_n \sinh\frac{\beta\omega^+_n}{2}}{2\prod\limits_{\substack{m\\ m\neq n} }\qty[1-\qty(\frac{\omega^+_n}{\omega_m^+})^2]},
\end{equation}
and, hence, the residues of $G(\omega)$,
\begin{equation}
    \text{Res}_{\omega_n^+}G(\omega)=2ir_n^+.
\end{equation}

\subsection{The spectral duality relation}
\label{subapp:SDR}
To derive the spectral duality relation, we have to assume the power-law asymptotics of $G(\omega)$ and the thermal product formula \eqref{app:tpf}. Using the product decomposition of $G(\omega)$ and the definition of $\rho(\omega)$, we get
\begin{align}
    \rho(\omega)&=\frac{G(0)}{2i}\lim_{N\rightarrow\infty}\qty(\frac{\pi^-_N(\omega)}{\pi^+_N(\omega)}-\frac{\pi^-_N(-\omega)}{\pi^+_N(-\omega)}) \nonumber\\
    &=\frac{G(0)}{2i}\lim_{N\rightarrow\infty}\frac{\pi_N^+(-\omega)\pi_N^-(\omega)-\pi_N^+(\omega)\pi_N^-(-\omega)}{\pi^+_N(\omega) \pi^+_N(-\omega)}.
\end{align}
Using the definition
\begin{equation}
    S(\omega)=\lim_{N\rightarrow\infty}\pi_N^+(\omega)\pi_N^-(-\omega),
\end{equation}
and the convergence of
\begin{equation}
    \lim_{N\rightarrow\infty}\pi_N^+(\omega)\pi_N^+(-\omega)=\prod_n \qty[1-\qty(\frac{\omega}{\omega_n^+})^2],
\end{equation}
along with the thermal product formula \eqref{app:tpf}, we get the spectral duality relation \cite{Grozdanov:2024wgo},
\begin{equation}
    S(\omega)-S(-\omega)=2i\lambda \sinh\frac{\beta\omega}{2},
\end{equation}
with 
\begin{equation}
    \lambda=-\frac{\mu}{G(0)}.
\end{equation}

We can now also consider the residues of $-1/G(\omega)$, corresponding to the respective poles $\qty{\omega_n^-}$, expressed as
\begin{equation}
    \text{Res}_{\omega_{n}^-}\qty(-\frac{1}{G(\omega)})=2i r_n^-=-\frac{1}{G'(\omega_n^-)}.
\end{equation}
By using the product decomposition (see Eq.~\eqref{app:G-product}) and the spectral duality relation, we can write
\begin{equation}
    \frac{1}{2i}\qty(-\frac{1}{G(\omega)}+\frac{1}{G(-\omega)})=\frac{\mu}{G(0)^2}\lim_{N\rightarrow\infty}\frac{\sinh\frac{\beta\omega}{2}}{\pi_N^-(\omega)\pi_N^-(-\omega)}. \label{app:tpf-zeroes}
\end{equation}
Essentially, this means that the thermal product formula holds for $-1/G(\omega)$ as well. Note that this does not directly imply the order-one property of the zeroes as the product in Eq.~\eqref{app:tpf-zeroes} does not necessarily correspond to an absolutely convergent canonical Hadamard product. Nevertheless, the order-one property \eqref{app:th:limsupT} likely holds for the zeroes as well, and is guaranteed when we can interpret the zeroes as a spectrum of a deformed theory for which the thermal product hypothesis \eqref{app:hypothesis} applies.

\section{Generalised Christmas tree configuration}
\label{app:xmas}
In this appendix, we provide the technical details that underlie the results of Section~\ref{subsec:christmas}. 

\subsection{A preliminary theorem}
When considering correlation functions that exhibit a Christmas-tree-like spectrum, their asymptotic behaviour is determined solely by a small number of parameters that determine the angle and the spacing of the asymptotic lines. Importantly, the small-$\omega$ details of the spectrum are largely irrelevant for the large-$\omega$ behaviour of the correlators. This is summarised in the following theorem, which states that we can, at the level of a logarithmic derivative, substitute the locations of poles and zeroes with their asymptotic values, with errors that fall off faster than $1/\omega$.
\begin{theorem}
\label{app:xmas-theorem}
    Let $w_m=d (m+\xi)$ for some complex $d=\abs{d}e^{i\theta}$ and $\xi$, and let $\delta \omega_n = \CO(m^{-\epsilon})$, for some $\epsilon>0$. For the sum
\begin{equation}
    \Delta(\omega)=\sum_{m=1}^\infty\frac{1}{\omega-w_m}-\frac{1}{\omega-w_m-\delta\omega_m} ,\label{app:delta_sum}
\end{equation}
the following holds:
\begin{equation}
\varphi \neq \theta: \quad \lim_{\abs{\omega}\rightarrow\infty}\abs{\omega}\Delta(\abs{\omega}e^{i\varphi})=0.
\end{equation}
\end{theorem}
\begin{proof}
Let $\omega=\abs{\omega}e^{i\varphi}$, $w_m=\abs{w_m}e^{i\chi_m}$, and let
\begin{equation}
    s_m(\omega)=\frac{1}{\omega-w_m}-\frac{1}{\omega-w_m-\delta \omega_m}.
\end{equation}
We want to show that
\begin{equation}
\varphi \neq \theta: \quad \lim_{\abs{\omega}\rightarrow\infty}\sum_{m=1}^\infty \abs{\omega} s_m(\omega)=0.
\end{equation}
It will be convenient to split the sum into two parts,
\begin{equation}
    \sum_{m=1}^\infty \abs{\omega}s_m(\omega)=\sum_{m=1}^{m_*-1} \abs{\omega}s_m(\omega)+\sum_{m=m_*}^\infty \abs{\omega}s_m(\omega), \label{app:sum-mstar}
 \end{equation}
for some large enough $m_*$, so that
\begin{equation}
    \frac{\abs{\delta\omega_{m}}}{\abs{w_{m}}\abs{\sin(\varphi-\chi_{m})}}<1 \text{ and } \abs{\varphi-\theta}>\abs{\chi_m-\theta}\text{, for all }m>m_*. \label{app:mstar}
\end{equation}
The finite sum in Eq.~\eqref{app:sum-mstar} trivially limits to zero in the $\abs{\omega}\rightarrow\infty$ limit. Next, we write
\begin{equation}
    \abs{\omega} \abs{s_m(\omega)}=\abs{\frac{\omega\delta\omega_m}{(\omega-w_m)^2}}{\abs{1-\frac{\delta\omega_m}{\omega-w_n}}}^{-1}.
\end{equation}
The first factor in this product can be estimated as
\begin{equation}
    \abs{\frac{\omega\delta\omega_m}{(\omega-w_m)^2}}\leq \frac{\abs{\delta\omega_m}}{2\abs{w_m}\abs{1-\cos(\varphi-\chi_m)}}.
\end{equation}
For the second factor, we can write
\begin{equation}
    {\abs{\frac{\delta\omega_m}{\omega-w_n}}}\leq \frac{\abs{\delta \omega_m}}{\abs{w_m}\abs{\sin(\varphi-\chi_m)}}.
\end{equation}
The conditions in Eq.~\eqref{app:mstar} tell us that we can find a sufficiently large constant $C$, such that
\begin{equation}
    \abs{\omega} \abs{s_m(\omega)}< C\abs{\frac{{\delta\omega_m}}{{w_m}}}=\CO(m^{-1-\epsilon}).
\end{equation}
The sum in Eq.~\eqref{app:delta_sum} can therefore be bounded by a convergent sum of $\omega$-independent summands,
\begin{equation}
     \abs{\Delta(\omega)}<C\sum_{m=1}^\infty \abs{\frac{{\delta \omega_m}}{{w_m}}}<\infty,
\end{equation}
which implies that we can use the dominated convergence theorem to exchange the order of the limit and the sum:
\begin{equation}
   \varphi \neq \theta: \quad \lim_{\abs{\omega}\rightarrow\infty}\sum_{m=1}^\infty \abs{\omega} s_m(\omega)=\sum_{m=1}^\infty  \lim_{\abs{\omega}\rightarrow\infty}\abs{\omega}s_m(\omega) =0.
\end{equation}
\end{proof}

\subsection{OPE constraints}
\label{subapp:xmas-OPE}
Here, we reiterate the results presented in Ref.~\cite{Dodelson:2023vrw}, and generalise them to an arbitrary number of Christmas tree branches. The thermal product formula \eqref{app:tpf}, in conjunction with the large-$\omega$ expansion of $\rho(\omega)$ \eqref{app:rhoOPE}, tells us
\begin{equation}
\label{app:product-OPE}
    \omega \rightarrow \infty^+: \quad \prod_n \qty[1-\qty(\frac{\omega}{\omega_n^+})^2] \sim e^{\beta\omega/2}\omega^{2\Delta-d}.
\end{equation}
Taking the logarithmic derivative, we get
\begin{equation}
    \omega \rightarrow \infty^+: \quad \sum_n \frac{1}{\omega-\omega_n^+}+\frac{1}{\omega+\omega_n^+} \sim \frac{\beta}{2}+\frac{2\Delta-d}{\omega}+\ldots.
\end{equation}
Assuming the Christmas tree form of the spectrum
\begin{equation}
    \omega^+_{m,j}=d^+_j (m+\xi^+_j)+\CO(m^{-\epsilon}),
\end{equation}
we can write
\begin{equation}
    \omega \rightarrow \infty^+: \quad \sum_j \sum_{m=1}^\infty \frac{1}{\omega-d_j^+(m+\xi_j^+)}+\frac{1}{\omega+d_j^+(m+\xi_j^+)}\sim \frac{\beta}{2}+\frac{2\Delta-d}{\omega}+\ldots. \label{app:dsum}
\end{equation}
Here, we have used Theorem~\ref{app:xmas-theorem} to conclude that the $\CO(m^{-\epsilon})$ correction does not affect the first two terms in the asymptotic expansion. The sum in Eq.~\eqref{app:dsum} can be evaluated explicitly in terms of the digamma function $\psi$ as
\begin{equation}
    \sum_{m=1}^\infty \frac{1}{\omega-d_j^+(m+\xi_j^+)}+\frac{1}{\omega+d_j^+(m+\xi_j^+)} = \frac{\psi\qty(1+\xi_j^+-\frac{\omega}{d_j^+})-\psi\qty(1+\xi_j^++\frac{\omega}{d_j^+})}{d_j^+},
\end{equation}
and expanded asymptotically using Stirling's approximation
\begin{equation}
    \psi(\omega) \sim \ln\omega-\frac{1}{2\omega}+\CO(\omega^{-2}), \label{app:stirling}
\end{equation}
where the branch cut of the logarithm runs along the negative real axis. By using the fact that $\Im d_j^+<0$, we obtain
\begin{equation}
    -\sum_j \frac{i\pi}{d_j^+}+\frac{1+2\xi_j^+}{\omega}= \frac{\beta}{2}+\frac{2\Delta-d}{\omega}.
\end{equation}
which recovers Eqs.~\eqref{eq:christmas_constraint} and \eqref{eq:xmas_constraint}.

We remark on the generalisation of the above configuration, where the lines of poles scale as
\begin{equation}
    \omega^+_{m,j}\sim m^{\alpha_j}.
\end{equation}
The condition \eqref{app:n_limsup} makes it clear that we must have $\alpha_j\geq 1$ with at least one $\alpha_j=1$. We can then write
\begin{equation}
    \prod_n \qty[1-\qty(\frac{\omega}{\omega_n^+})^2]=\prod_j f_j(\omega),
\end{equation}
where
\begin{equation}
    f_j(\omega)=\prod_{m=1}^\infty \qty[1-\qty(\frac{\omega}{\omega_{m,j}^+})^2].
\end{equation}
By Borel's theorem (see Eq.~\eqref{app:th:Borel}), the function $f_j(\omega)$ is of order $1/\alpha_j$. As such, it is asymptotically bounded as $\ln{\abs{f_j(\omega)}}<\abs{\omega}^{1/\alpha_j+\epsilon}$, for any $\epsilon>0$. For $\alpha_j>1$, the function $f_j(\omega)$ cannot contribute to the leading exponential factor in the asymptotic expansion in Eq.~\eqref{app:product-OPE}. We therefore find that 
\begin{equation}
    -i\sum_{\substack{j\\\alpha_j=1}}\frac{1}{d_j^+}=\frac{\beta}{2\pi}.
\end{equation}
No simple generalisation exists for the subleading term involving the scaling dimension. This is due to the fact that the inclusion of the $\alpha_j>1$ branches can be shown to require an asymptotic expansion which is not of the form of $\omega_{m,j}^+\sim m^{\alpha_j}+\CO(m^{\alpha_j-1})$ (see Ref.~\cite{Dodelson:2023vrw}).

\subsection{Zeroes, poles, and asymptotics}
\label{subapp:xmas-retarded}
We now consider the power-law asymptotics of $G(\omega)$ (see Appendix \ref{subapp:powerlaw}) in the context of the Christmas tree configurations. Specifically, we consider the lines of poles and zeroes that have the form
\begin{subequations}
\label{app:xmas-branches}
\begin{align}
    \omega_{m,j}^+&=d^+_j (m+\xi_j^+)+\CO(m^{-\epsilon}),\\
    \omega_{m,j}^-&=d^-_j (m+\xi_j^+)+\CO(m^{-\epsilon}).
\end{align}    
\end{subequations}
For simplicity, we consider non-overlapping lines, i.e., with
\begin{equation}
\label{app:non-overlapping}
    \Im \frac{d_j^+}{d_{j'}^+}\neq 0 \text{ and } \Im \frac{d_j^-}{d_{j'}^-}\neq 0 \text{, for } j\neq j'.
\end{equation}
The adaptation of the Hadamard's decomposition theorem for meromorphic functions (see e.g.~Ref.~\cite{Goldberg2008-ya}) says that we can write
\begin{equation}
G(\omega)=e^{h(\omega)}\frac{\prod\limits_{j}\prod\limits_{m=1}^\infty e^{\omega/\omega_{m,j}^-}\qty(1-\frac{\omega}{\omega_{m,j}^-})}{\prod\limits_{j}\prod\limits_{m=1}^\infty e^{\omega/\omega_{m,j}^+}\qty(1-\frac{\omega}{\omega_{m,j}^+})},\label{app:G-Hadamard}
\end{equation}
for some entire $h(\omega)$. This decomposition differs from the product expansion \eqref{app:G-product} by the presence of the exponential factors, and the fact that both the numerator and the denominator converge individually (and absolutely). Taking the logarithmic derivative of Eq.~\eqref{app:G-Hadamard}, we get
\begin{equation}
    \frac{G'(\omega)}{G(\omega)}=h'(\omega) + \sum_j \sum_{m=1}^\infty \qty[\frac{1}{\omega_{m,j}^-}+ \frac{1}{\omega-\omega_{m,j}^-}]-\sum_j \sum_{m=1}^\infty \qty[\frac{1}{\omega_{m,j}^+} +\frac{1}{\omega-\omega_{m,j}^+}]. \label{app:logG-sum}
\end{equation}
The assumption about the power-law asymptotics tells us that $G'(\omega)/G(\omega)$ must decay as $1/\abs{\omega}$ in all of the relevant directions of the complex plane. We can use Theorem~\ref{app:xmas-theorem} to substitute $\omega_n^\pm \rightarrow d^\pm_j(m+\xi_j^\pm)$ and evaluate the sums in terms of the digamma function. Using Stirling's approximation \eqref{app:stirling}, the yields, to leading order,
\begin{equation}
    \frac{G'(\omega)}{G(\omega)}\sim h'(\omega)+ \sum_j\frac{1}{d_j^-}\ln\qty(-\frac{\omega}{d_j^-})-\sum_j\frac{1}{d_j^+}\ln\qty(-\frac{\omega}{d_j^+})+\ldots. \label{app:dlog-G}
\end{equation}
First, we conclude that $h'(\omega)=\text{const}.$, since any other terms would lead to diverging asymptotics. Next, we introduce\begin{equation}
    \omega=\abs{\omega}e^{i\varphi}, \quad d_j^\pm=\abs{d_j^\pm}e^{i\theta_j^\pm},
\end{equation}
and rewrite Eq.~\eqref{app:dlog-G} as
\begin{equation}
    \frac{G'(\omega)}{G(\omega)}\sim v(\varphi)\ln\abs{\omega} +\ldots,
\end{equation}
and demand that $v(\varphi)= 0$ for all $\varphi \neq \theta_j^\pm$. We emphasise again that the branch cut of the logarithm must be taken along the negative real axis. Consequently, $v(\varphi)$ has discontinuities at
\begin{equation}
    \varphi = \theta_j^\pm,
\end{equation}
with their respective sizes given by
\begin{equation}
    \lim_{\epsilon\rightarrow 0}\qty[v\qty(\varphi+\epsilon)-v\qty(\varphi-\epsilon)]=2\pi i\qty[\sum_{\substack{j\\\theta_j^+=\varphi}}\frac{1}{d_j^+}-\sum_{\substack{j\\\theta_j^-=\varphi}}\frac{1}{d_j^-}].
\end{equation}
Taking Eq.~\eqref{app:non-overlapping} into account, the only way to cancel the discontinuities is by having the same number of branches of poles and zeroes, and setting
\begin{equation}
    d_j \equiv d_j^+=d_j^-. \label{app:parallel-condition}
\end{equation}
This automatically gives $v(\varphi)=0$.

We can now rewrite the sum \eqref{app:logG-sum} as
\begin{equation}
    \frac{G'(\omega)}{G(\omega)}=\qty[h'(\omega)+\sum_j \sum_{m=1}^\infty \frac{1}{\omega_{m,j}^-}-\frac{1}{\omega_{m,j}^+}]+\qty[\sum_j \sum_{m=1}^\infty \frac{1}{\omega-\omega_{m,j}^-}-\frac{1}{\omega-\omega_{m,j}^+}],
\end{equation}
where the condition \eqref{app:parallel-condition} ensures the convergence of the rearranged sum. We can again use Theorem \ref{app:xmas-theorem} to approximate $\omega_{m,j}^\pm$ with their asymptotic values. The sum then evaluates to a sum of digamma functions, which can again be expanded asymptotically to give
\begin{equation}
    \qty[\sum_j \sum_{m=1}^\infty \frac{1}{\omega-\omega_{m,j}^-}-\frac{1}{\omega-\omega_{m,j}^+}] \sim \sum_j \frac{\xi_j^+-\xi_j^-}{\omega}+\ldots.
\end{equation}
This sets
\begin{equation}
\label{app:f-sol}
    h(\omega)=\ln G(0)+\omega\qty[\sum_j \sum_{m=1}^\infty \frac{1}{\omega_{m,j}^+}-\frac{1}{\omega_{m,j}^-}],
\end{equation}
resulting in
\begin{equation}
    \varphi \neq \theta_j: \quad \frac{G'(\omega)}{G(\omega)} \sim \frac{\sigma}{\omega}+\ldots, \label{app:log-power}
\end{equation}
where
\begin{equation}
    \sigma = \sum_j \xi_j^--\xi_j^+.
\end{equation}
Since there is no non-trivial phase dependence in Eq.~\eqref{app:log-power}, we immediately obtain
\begin{equation}
    \varphi \neq \theta_j: \quad G(\omega) \sim \omega^{-\sigma}.
\end{equation}
Note that it is now clear that the Hadamard decomposition \eqref{app:G-Hadamard} is equivalent to the product decomposition \eqref{app:G-product}, since we can use Eq.~\eqref{app:f-sol} to write
\begin{equation}
    G(\omega)=G(0)\prod_j \prod_{m=1}^\infty\frac{1-\frac{\omega}{\omega_{m,j}^-}}{1-\frac{\omega}{\omega_{m,j}^+}}.
\end{equation}

\acknowledgments
We thank Richard Davison, Giorgio Frangi, Elizabeth Helfenberger, Robin Karlsson, Mark Mezei, Jaka Pelai\v c, and Alexander Zhiboedov for many elucidating and fruitful discussions. The work of S.G. was supported by the STFC Ernest Rutherford Fellowship ST/T00388X/1. The work is also supported by the research programme P1-0402 and the project J7-60121 of Slovenian Research Agency (ARIS). M.V. is funded by the STFC Studentship ST/X508366/1.


\bibliographystyle{JHEP}
\bibliography{biblio.bib}

\end{document}